\documentclass[journal,10pt]{IEEEtran}

\ifCLASSINFOpdf
  \usepackage[pdftex]{graphicx}
  \usepackage{epstopdf}
  \DeclareGraphicsExtensions{.eps, .pdf, .jpeg}
\else
\fi

\ifCLASSINFOpdf
	\usepackage[pdftex]{graphicx}
\else
	\usepackage[dvips]{graphicx}
\fi

\usepackage{amsmath}
\usepackage{amsthm}
\usepackage{cite}
\usepackage{amsfonts}
\usepackage{array}
\usepackage{booktabs}
\usepackage{colortbl}
\usepackage{mdwmath}
\usepackage{mdwtab}
\usepackage{eqparbox}
\usepackage[ruled,vlined]{algorithm2e}

\usepackage{amssymb}
\usepackage{amsthm}

\newtheorem{theorem}{Theorem}

\newtheorem{lemma}{Lemma}
\newtheorem{corollary}{Corollary}

\newtheorem{remark}{Remark}
\newtheorem{assumption}{Assumption}

\usepackage{csquotes}
\usepackage[dvipsnames]{xcolor}
\usepackage{color}
\usepackage{color, soul}
\usepackage{xr}

\usepackage{array,ragged2e}
\usepackage{multirow}
\usepackage{hhline}
\usepackage{array}
\usepackage{booktabs}
\usepackage{colortbl}
\usepackage{mdwmath}
\usepackage{mdwtab}

\usepackage{footnote}
\makesavenoteenv{tabular}
\makesavenoteenv{table}

\newcommand{\figSize}[1]{\def\@figSize{#1}}
\figSize{3.11in} 

\begin{document}
\title{Geo-Location Based Access for Vehicular Communications: Analysis and Optimization via Stochastic Geometry}
%
\author{Francisco~J.~Martin-Vega, Beatriz~Soret, Mari~Carmen~Aguayo-Torres, \\Istvan~Z.~Kovacs and Gerardo~Gomez
\thanks{F.~J.~Martin-Vega, M.~C.~Aguayo-Torres and G.~Gomez are with the Departamento
de Ingenier\'ia de Comunicaciones, Universidad de M\'alaga, M\'alaga 29071, Spain (e-mail: \{fjmvega, aguayo, ggomez\}@ic.uma.es)}
\thanks{B.~Soret and I.~Z.~Kovacs are with Nokia Bell Labs, Aalborg, Denmark.}
\thanks{This work has been supported by the Spanish Government (Ministerio de Econom\'ia y Competitividad), under grants TEC2013-44442-P and TEC2016-80090-C2-1-R, the Universidad de M\'alaga and Nokia Bell Labs.}
%
%
\thanks{A small part of this manuscript has been submitted to Vehicular Technology Conference, spring  2017 \cite{Martin-Vega16b}.}
\thanks{This work has been submitted to the IEEE for possible publication. Copyright
may be transferred without notice, after which this version may no longer be
accessible.}}
\markboth{IEEE Transactions on Vehicular Technology,~Vol.~XX, No.~XX, XXX~2016}
{}

\maketitle

\begin{abstract}
Delivery of broadcast messages among vehicles for safety purposes, which is known as one of the key ingredients of Intelligent Transportation Systems (ITS), requires an efficient Medium Access Control (MAC) that provides low average delay and high reliability. To this end, a Geo-Location Based Access (GLOC) for  vehicles has been proposed for Vehicle-to-Vehicle (V2V) communications, aiming at maximizing the distance of co-channel transmitters while preserving a low latency when accessing the resources. In this paper we analyze, with the aid of stochastic geometry, the delivery of periodic and non-periodic broadcast messages with GLOC, taking into account path loss and fading as well as the random locations of transmitting vehicles. Analytical results include the average interference, average Binary Rate (BR), capture probability, i.e., the probability of successful message transmission, and Energy Efficiency (EE). Mathematical analysis reveals interesting insights about the system performance, which are validated thought extensive Monte Carlo simulations. In particular, it is shown that the capture probability is an increasing function with exponential dependence with respect to the transmit power and it is demonstrated that an arbitrary high capture probability can be achieved, as long as the number of access resources is high enough. Finally, to facilitate the system-level design of GLOC, the optimum transmit power is derived, which leads to a maximal EE subject to a given constraint in the capture probability. 
\end{abstract}


%
\IEEEpeerreviewmaketitle

\section{Introduction}
\label{sec:Introduction}
\IEEEPARstart{E}{xisting}
 Intelligent Transportation Systems (ITS) mainly rely on cameras, roadside sensors, and variable-message signs to monitor and control the traffic. The high operational cost, limited effectiveness and reliability of such an approach have lead to a growing interest in Vehicular Ad Hoc Networks (VANETs) \cite{Festag14}, \cite{Karagiannis11}, which allows to implement more proactive tools for safety and traffic efficiency applications. The development of Global Positioning System (GPS) and other related techniques \cite{ Chen15} to achieve fine position tracking have also played a crucial role in the expectations that have been put on VANETs. 
 
Within the safety applications that require direct Vehicle-to-Vehicle (V2V) communication, two types of broadcast messages can be identified: periodic and non-periodic messages \cite{Chen10}, \cite{Hartenstein08}.  Periodic messages aim at achieving vehicle awareness and consist on the periodic transmission of broadcast status messages, informing nearby vehicles of their position, velocity and direction. 
On the other hand, non-periodic, e.g. event-driven messages, are transmitted to respond to specific hazardous situations, like a sudden hard braking vehicle from other neighboring vehicles, the presence of emergency vehicles, or to undertake early countermeasures to prevent chain-reaction accidents \cite{Biswas06}. 

These broadcast messages require high reliability and low latency, thus making the design of the Medium Access Control layer (MAC) an issue of paramount importance. On the one hand, centralized solutions for MAC are associated with high reliability since collisions can be reduced. However, such solutions have two main drawbacks: (i) they are complex as they require an association procedure, control channels, and infrastructure; and (ii) they have an inherent mean delay. For these reason distributed solutions, like slotted ALOHA and Carrier Sense Multiple Access (CSMA) \cite{ElSawy13b}, are normally preferred in this scenario. 
%
Collisions are a limiting factor related to slotted ALOHA. Hence, ALOHA schemes evolved into listen-before-talk solutions, like CSMA, aiming to avoid collisions of nearby transmitting nodes. 
In this context, IEEE 802.11p, which makes use of CSMA as a MAC protocol, has been presented as an attractive solution for VANETs since it provides decentralized and ad hoc connectivity. 

Nevertheless, 802.11p suffers from the main limitations related to 802.xx standards, such as poor scalability to high traffic density and poor support of high mobility \cite{Hameed14}. Hence, solutions based on 4G and 5G cellular networks come to the fore \cite{Lee16}. LTE V2X is the response of the 3GPP standardization body to the high market expectations and will use the same principles as those that are envisioned for Device-to-Device (D2D) communications. 
Here, instead of the traditional communication through the infrastructure, which may suffer from long delays, the local data exchange through the direct V2V path is preferred. With autonomous mode of operation, it is even possible for the devices to select the transmission resources without network involvement. Yet, the network has a key role in providing time and frequency synchronization.

\subsection{Motivation and Related Work}
A novel MAC technique that makes use of the cellular network for synchronization purposes is Geo-Location based access (GLOC) \cite[Sect. 23.14.1.1. (support for V2X sidelink)]{3gpp2016b}, \cite{Soret16}. In this technique, vehicles access the channel based on its position. The road is divided into segments, where each segment is associated with a single orthogonal Access Resource (AR). The mapping between ARs and segments is made to maximize the co-channel distance.
In \cite{Soret16}, it is proposed GLOC for the discovery and communication phase in V2V communications. It is shown through simulation that the benefits of this technique are: (i) high reliability, since the distance to interfering vehicles can be increased with the number of ARs and (ii) MAC layer does not add any delay on accessing the channel, i.e., vehicles start their transmission once they have data to transmit.

Besides simulation based studies like \cite{Soret16}, analytical models can provide further insights about the inter-plays among reliability and binary rate as well as the number of ARs for medium access, or the transmit power. Additionally, mathematical analysis leads to expressions that can be evaluated quickly and allows to perform optimization of the most relevant performance metrics. 
Here is where stochastic geometry \cite{ElSawy13a, Haenggi09, Haenggi13} appears as a promising tool, since it allows tractable and realistic analysis due to the random nature of the location of transmitting vehicles in VANETs. 
For instance, \cite{Blaszczyszyn09} analyzes, with the aid of stochastic geometry, the capture probability, average throughput and mean density progress of transmitted packets for the case of unicast transmissions with ALOHA.  
In \cite{Farooq15}, CSMA for unicast multi-hop communications is considered with several routing strategies. It also considers multi-lane abstraction model which is more accurate than single-lane models for wide roads. 
The case of a head vehicle that broadcasts info and control messages to a sectorized cluster of client vehicles is considered in \cite{Jeong13a}. This work models the positions of vehicles as a Cox process whose density follows a Fox distribution; however, the interference caused by other transmitting vehicles is not taken into account. 
The spatial propagation of broadcast information is tackled in \cite{Wu09}; nevertheless, the signal propagation is neglected and it is assumed that transmission is always successful as long as the distance towards the receiver is smaller than a given distance. 
The performance of IEEE 802.11p is assessed with the aid of stochastic geometry and queuing theories in \cite{Tong16}. Here it is considered the temporal behavior of CSMA which adds a delay to access the system by means of a back-off counter. To account for the spatial dependence, which is derived from the carrier sensing, and also for the temporal behavior, which is derived from the back-off counter, a discrete 
Mat\'ern Hard Core Point Process (HCPP) is proposed to model the locations of concurrent transmitters. 

\subsection{Main Contributions}

In this work, GLOC access technique is analyzed taking into account the velocity-dependent safe distance, $d_{\rm safe}$, between vehicles of the same lane. Such a safe distance imposes some correlation between locations of the vehicles, since there is not two neighboring vehicles closer than $d_{\rm safe}$. Hence, the location of vehicles in this paper is modeled by means of a Mat\'ern HCPP of type II. However, such a point process is generally intractable, and only some moments of the interference can be obtained without resorting to approximations \cite{Haenggi11}. To overcome such an intractability, we will use conditional thinning as in \cite{Martin-Vega16}. In simple terms, the locations of vehicles are first assumed to be placed according to a PPP of a given density. Then, spatial constrains (correlation) in the form of a minimum distance between points are imposed by means of an indicator function, but only in the proximity of the transmitter and the receiver. Additionally, it is considered that the length of the road is much higher than its width, and hence it is assumed that locations of the vehicles in each lane can be modeled as points in the real line. 
Based on these modeling assumptions, which are validated against extensive Monte Carlo simulations, we provide the following contributions:
 
1) \textit{Mathematical framework for analysis of geo-location based access}: We propose a mathematical framework for the analysis of GLOC considering a minimum distance between vehicles of the same lane. Two kind of resource allocation schemes are considered: Single-Lane Partition (SLP) and Multi-Lane Partition (MLP), which have different trade-offs and mainly differ on whether 
lane-finding is required or not. With SLP, the road is divided in different segments, whereas with MLP, each lane is divided in segments. Both broadcast messages, i.e., periodic and non-periodic, are modeled to obtain a complete understanding about the capabilities of GLOC as a MAC for ITS. 
Additionally, system-level parameters like message size, reporting rate, broadcast distance, etc. are taken from recommentations of the 3GPP Work Items \cite{3gpp2016} and \cite{3gpp2015} to study the support of LTE for V2V services. 
The path loss slope and path loss exponent is taken from \cite{Karedal11} where it has been performed a vast V2V channel measurement campaign conducted in Sweden over a carrier frequency of $5.2$ GHz. 
Interestingly, the path loss exponent in V2V channels, $\alpha$, is normally smaller than $2$ \cite{Karedal11, Viriyasitavat15,Fernandez13}. This means that only one-dimensional PPPs can be considered%
\footnote{As it is mentioned in \cite{Haenggi13}, the Probability Generating Functional of the PPP in $\mathbb{R}^d$, with $d \in \mathbb{N}^+$, only exists for a path loss exponent, $\alpha>d$. Hence, if we consider two-dimensional PPPs, the mathematical analysis is restricted to the case $\alpha>2$.}. 
Finally, mathematical expressions for a wide variety of performance indicators  have been obtained, leading to a deep understanding of the studied techniques. In particular, the capture probability, the average interference, the average Binary Rate (BR) and the average Energy Efficiency (EE) are derived. 

2) \textit{Theoretical insights}: Many useful insights have been obtained from the derived expressions. Interestingly, it has been shown that: (i) the capture probability is an increasing function with respect to the transmit power with exponential dependence; (ii) the system is noise-limited for MLP when the number of ARs is high enough whereas it is interference-limited in case of SLP; (iii) the average interference diverges when it is evaluated in co-channel segments with SLP, whereas it always converges for the case of MLP. The fact that with MLP the system is noise-limited for a given number of ARs means that it is possible to achieve an arbitrary high capture probability by increasing the transmit power. 

3) \textit{Optimization}: The optimum transmit power that achieves maximal EE subject to a minimum capture probability is obtained. Such a minimum value is expressed as a percentage, $\delta$, of the maximum capture probability that can be achieved. Interestingly, the same optimal transmit power is obtained for SLP and MLP.

The rest of the paper is organized as follows. Section \ref{sec:System Model} presents the system model. The mathematical analysis and optimization are explained in Section \ref{sec:Performance Analysis}. Finally, numerical results are illustrated in Section \ref{sec:Numerical Results} whereas conclusions are drawn in Section \ref{sec:Conclusion}. 


\section{System Model}
\label{sec:System Model}

\begin{figure}[t]
\centering
\includegraphics[width=3.5in]{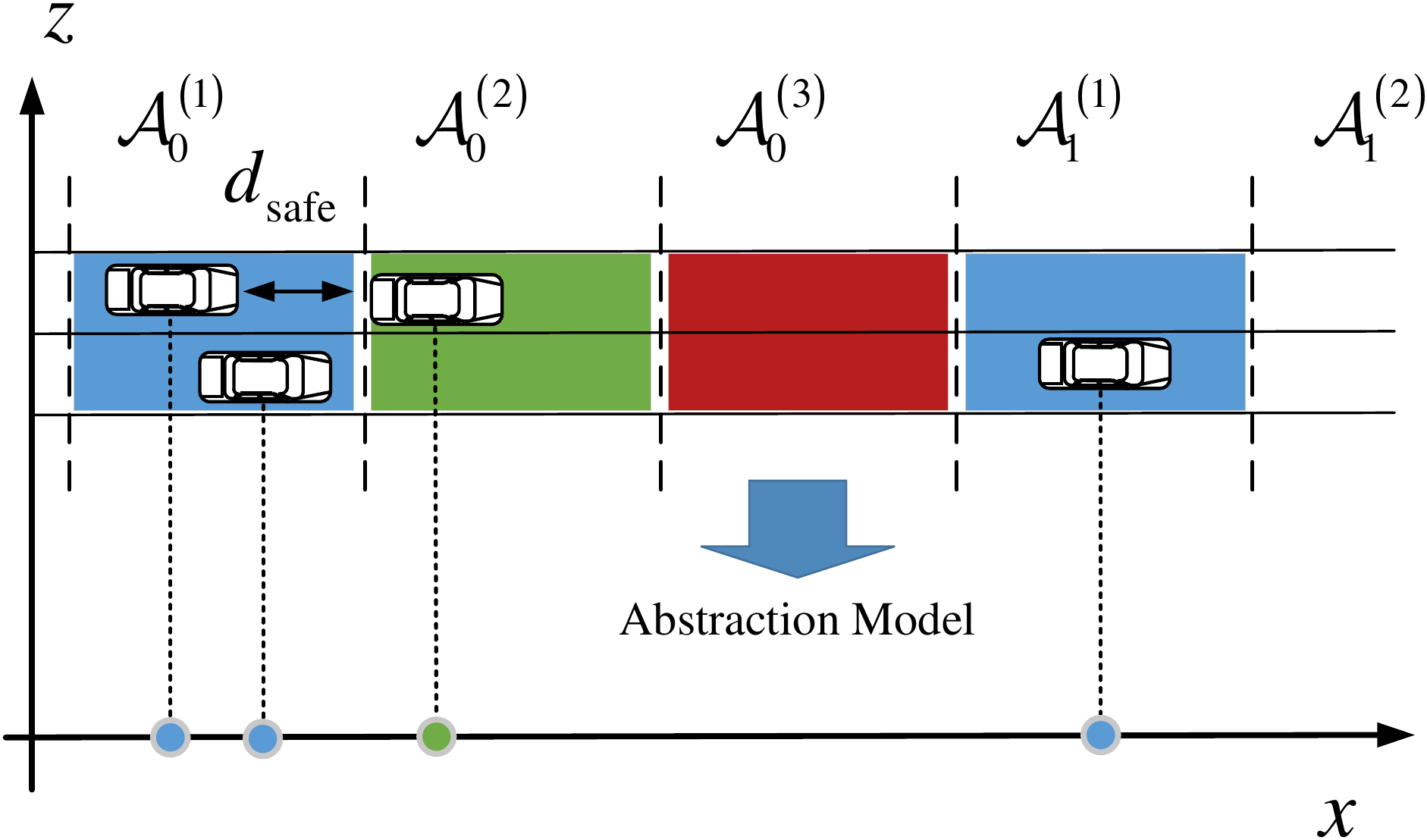}
\caption{Sketch of GLOC based access with SLP for a two-lane road and $3$ ARs. Different colors are associated to different ARs. On the bottom it is shown the abstraction model for the study of the road where positions of vehicles are treated as points in the real line.}
\label{fig:SLP}
\end{figure}

A straight road with $n_L$ lanes is considered as appears in figures \ref{fig:SLP} and \ref{fig:MLP} ($n_L=2$), where the length of the road is much greater than its width and thus the $z$ coordinate can be neglected. 
As pointed out in \cite{3gpp2016}, there is a velocity-dependent safe distance between vehicles of the same lane, referred to $d_{\rm safe}$. Hence, positions of vehicles within the same lane are assumed to follow a Mat\'ern HCPP of type II, $\Phi_L=\{\mathrm{V}_0, \mathrm{V}_1, \cdots\}\subset \mathbb{R} $, whose density is $\lambda_L$ and its minimum distance between points is $d_{\rm safe}$. The assumption of a minimum distance between vehicles leads to a maximum vehicle density per lane, which is 
$\lambda_{L, \rm max} = 1/\left( 2 d_{\rm safe}\right)$.
It is assumed that at a given time instant a vehicle has data to transmit with probability $p_\mathrm{a}$; hence the set of active vehicles $\Phi_L^{(\mathrm{a})}=\{\mathrm{VT}_0, \mathrm{VT}_1, \cdots\}$ is obtained through independent thinning from $\Phi_L$ with density $\lambda_L p_\mathrm{a}$. A summary of main symbols and functions is provided in Table \ref{tab:Symbols}.

\subsection{Resource Partition Schemes} 
With GLOC, the road is divided into segments of length $d_\mathcal{A}$ meters, and each segment is associated with a given orthogonal AR. The useful system bandwidth, $b_w$, is divided between the ARs. 
At a given time instant, each vehicle with data to transmit determines its current segment based on its position and then, it transmits with the mapped AR. The mapping between segments and ARs is made to maximize the co-channel distance.
Modeling the location of the vehicles randomly in terms of point processes allow us to treats the VANET as a snapshot of a stationary random field of communicating vehicles, where realizations of such point process are associated with different vehicle locations.
Hence, it is assumed that the segment does not change during the transmission of a single message. 
Besides, it is considered that the vehicles are aware of the mapping between segments and ARs, and the position and size of the segments. This can be achieved following the same process as specified in \cite[ Sect. 23.14.1.1. (support for V2X sidelink)]{3gpp2016b}. 
We propose two resource partition schemes, identified as SLP and MLP, that mainly differ on whether 
lane-finding is required or not. With SLP, the road is divided in different segments, whereas with MLP, each lane is divided in segments. 

The frequency allocation process of SLP is depicted as appears next:
\begin{enumerate}
\item The road is divided into segments of $d_{\cal A}$ meters. Each segment consists of $n_L$ lanes. 
\item A bandwidth of $b_w/n_{\rm AR}$ is allocated to each AR, where $n_{\rm AR}$ is the number of ARs.
\item The segments are grouped into consecutive clusters of $n_{\rm AR}$ segments. A single orthogonal AR is allocated to every segment within a given cluster. The mapping between segments and ARs is made with maximum co-channel distance criterion, aiming at minimizing the interference. 
\end{enumerate}
A sample of SLP scheme for $n_{\rm AR}=3$ and $n_L=2$ is shown in the top of Fig. \ref{fig:SLP}, whereas the mathematical abstraction model as a one-dimensional point process is illustrated at the bottom of the figure. In this case, each color (blue, green and red) represents a different AR whereas segments are represented as 
$\mathcal{A}_c^{(j)}$, where $j$ identifies the AR, and $c$ identifies the cluster. 

\begin{figure}[t]
\centering
\includegraphics[width=3.5in]{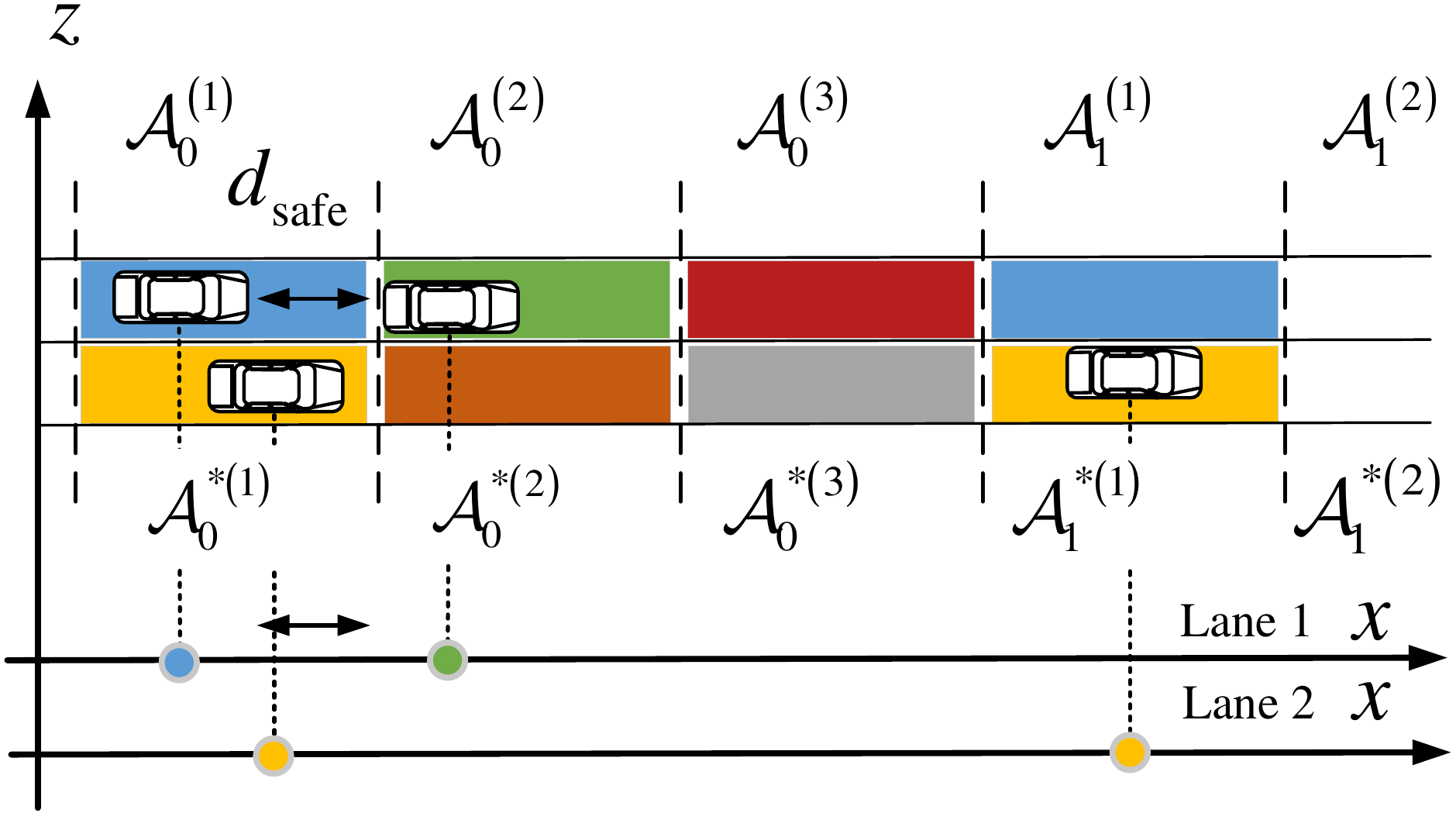}
\caption{Sketch of GLOC based access with MLP for a two-lane road and $6$ ARs. Different colors are associated to different ARs. On the bottom it is shown the abstraction model for the study of each lane where positions of vehicles are treated as points in the real line.}
\label{fig:MLP}
\end{figure}

On the other hand, MLP considers that each segment only contains a single lane. The process to allocate frequencies with MLP is described as follows:
\begin{enumerate}
\item The system bandwidth is equally divided among lanes. Therefore, there is $b_w/n_L$ Hz available for each lane and there is no interference among different lanes. 
\item Each lane is divided into segments of $d_{\cal A}$ meters. 
\item A bandwidth of $b_w/(n_{\rm AR} \cdot n_L)$ is allocated to each AR, where $n_{\rm AR}$ is the number of ARs per lane. Thus, the overall number of ARs is $n_L \cdot n_{\rm AR}$.
\item The segments of each lane are grouped in consecutive clusters of $n_{\rm AR}$ segments. A single orthogonal AR is allocated to every segment within a given cluster. The mapping among segments of the same lane and ARs is made with maximum co-channel distance criterion, aiming at minimizing the interference. 
\end{enumerate}
Fig. \ref{fig:MLP} illustrates a sample of MLP for $n_{\rm AR}=3$ and $n_L=2$. In this case 
$\mathcal{A}_c^{(j)}$ represents segments that belong to lane $1$ whereas $\mathcal{A}_c^{*(j)}$ represents segments related to lane $2$. The abstraction model for each lane is represented at the bottom of the figure. 

\begin{table}
\renewcommand{\arraystretch}{1.0}
\caption{Summary of main symbols and functions used throughout the paper.}
\label{tab:Symbols}
\centering
\begin{tabular}{ l l }
\toprule
{Symbol} & {Definition}  
\\ \toprule
$\lfloor \cdot \rfloor$, $\lceil \cdot \rceil$ & Floor and ceil functions
\\ \hline
$_2F_1(\cdot,\cdot,\cdot,\cdot)$ & Gauss hypergeometric function defined in \cite{Abramowitz65} (Ch. 15).
\\ \hline
$\mathbb{E}\left[  \cdot\right]$
 & Expectation operator
\\ \hline
$\Pr \left( \cdot  \right)$
 & Probability measure 
\\ \hline
${\bf{1}}\left( \cdot \right)$
 & Indicator function 
\\ \hline
$\mathfrak{b}_x (r)$ & Ball centered at $x$ with radius $r$ 
\\ \hline 
$n_L$ & Number of lanes 
\\ \hline
$d_{\rm safe}$ & Safe distance between vehicles 
\\ \hline
$d_{\rm max}$ & Maximum communication distance
\\ \hline
 $\Phi_L, \lambda_L$ & HCPP that models locations of vehicles within the \\ &  same lane and its density
\\ \hline
$\Phi^{(a)}_L, \lambda^{(a)}_L$ & Thinned HCPP that models locations of active vehicles 
\\ &  within the same lane and its density
\\ \hline
$d_{\cal A}$ & Segment's length
\\ \hline
$b_w$ & Useful system's bandwidth
\\ \hline
$n_{\rm AR}$ & With SLP it is the number of orthogonal resources. With
\\ &  MLP it is the number of orthogonal resources per lane
\\ \hline
${\cal A}_c^{(j)}$ & Segment associated with the $j$-th AR within cluster $c$
\\ \hline
${\cal A}^{(j)}$ & Union of all the segments associated with the $j$-th AR
\\ \hline
$\tau,\alpha$ & Slope and exponent of the path loss function
\\ \hline
$\rho_{\rm VT}$ & Transmit power per Hz
\\ \hline
$b_{\rm AR}$ & Bandwidth of a single AR
\\ \hline
$\Phi,\lambda$ & Point process that models the location  of vehicles in the 
\\ &   abstraction model for SLP and MLP
\\ \hline
$\Phi^{(a)},\lambda^{(a)}$ & Thinned point process that model the location of active
\\ &   vehicles in the abstraction model for SLP and MLP
\\ \hline
${\rm VT_0}, {\rm AR_0}$ & Probe vehicle transmitter and its related AR
\\ \hline
${\rm VT_i}, H_{\rm VT_i}$ & Generic active vehicle and its fading towards the \\ &  probe receiver
\\ \hline
$\sigma _n^2 ,I$ & Noise power
and aggregate interference 
\\ \bottomrule
\end{tabular}
\end{table}

Each scheme has different pros and cons. With SLP it is not required to identify the lane in which the vehicle is traveling, which relaxes the requirement imposed to position estimation. On the other hand, MLP considers that vehicles are capable of estimating their position and also their current lane; however, this can be achieved using similar techniques as proposed in \cite{Chen15}. Additionally, the requisites imposed over position estimation for 5G are around $30$ cm, which assures lane-awareness \cite{5G-PPP2015}.
On the negative side, it can be noticed that SLP leads to a higher density of co-channel interfering vehicles, since each segment has several lanes. Additionally, the minimum distance towards the nearest interfering vehicle is reduced, since in this case an interfering vehicle could be located in the same location as the receiver in a different lane. This does not happen in case of MLP thanks to the minimum (inter-vehicle) safe distance. Nevertheless the bandwidth for each AR is lower in MLP, since the bandwidth is also divided among lanes. 

These differences have also implications on the mathematical modeling of SLP and MLP. In particular, with SLP there is not a minimum distance between vehicles, and hence the position of interfering vehicles can be modeled as a PPP. 
The bandwidth per AR and vehicle densities are also different, as summarized in Table 
\ref{tab:SLP and MLP}.

The analysis is performed for a typical transmitter, i.e., a randomly selected Vehicle Transmitter (VT). This transmitter is named the probe VT, and it is represented with symbol $\mathrm{VT}_0$. In this paper we made an abuse of notation since $\mathrm{VT}_0$ is used to represent the probe VT as well as its position in the real line. Analogously, its associated AR is the probe AR, which is denoted by 
$\mathrm{AR}_0$. 
Symbol $\mathcal{A}_c^{(j)}$ identifies the segment associated with $j$-th AR within cluster $c$. 
The set that represents all the segments associated with AR $j$ is represented as 
$
\mathcal{A}^{(j)}=\bigcup\limits_{c=-\infty }^{\infty }{\mathcal{A}_{c}^{(j)}}
$.
Fig. \ref{fig:abstractionModel} shows a sketch of the abstraction model, either for SLP or MLP. In case of SLP, this abstraction model is related to a given road, whereas in case of MLP it is related to a given lane. 
Without loss of generality, it is considered that the probe segment, $\mathcal{A}_0^{(\mathrm{AR_0})}$, is centered at the origin. 
Being the probe segment centered at the origin, the $c$-th co-channel segment, $\mathcal{A}_{c}^{(\mathrm{AR}_{0})}$, can be expressed as

\begin{align}
& \mathcal{A}_{c}^{(\mathrm{AR}_{0})}= \left\{y\in \mathbb{R}:c n_{\mathrm{AR}} d_{\mathcal{A}}-\frac{d_{\mathcal{A}}}{2}\le y \right.
\left. <c n_{\mathrm{AR}}d_{\mathcal{A}}+\frac{d_{\mathcal{A}}}{2} \right\}
\end{align}

The $c$-th co-channel segment is centered around $c \cdot n_\mathrm{AR} \cdot d_\mathcal{A} $, with 
$c \in \mathbb{Z}$. 

\begin{table}
\renewcommand{\arraystretch}{1.3}
\caption{Modeling Differences between SLP and MLP}
\label{tab:SLP and MLP}
\centering
\begin{tabular}{ c c c c c}
\toprule
Scheme & $\lambda$ & $b_{\rm AR}$ & Minimum dist. & Type of $\Phi$\\
\toprule
SLP & $\lambda_L \cdot n_L$ & $b_w/n_{\rm AR}$ & $0$ & PPP \\
\hline
MLP & $\lambda_L$ & $b_w/\left( n_{\rm AR} n_L \right)$ & $d_{\rm safe}$ & HCPP\footnote{Although the position of vehicles with MLP are modeled  as a HCPP, it is approximated by means of conditional thinning as a PPP as it is stated with \textbf{Assumption \ref{ass:HCPP}}.} \\
\bottomrule
\end{tabular}
\end{table}

\begin{figure}[t]
\centering
\includegraphics[width=3.5in]{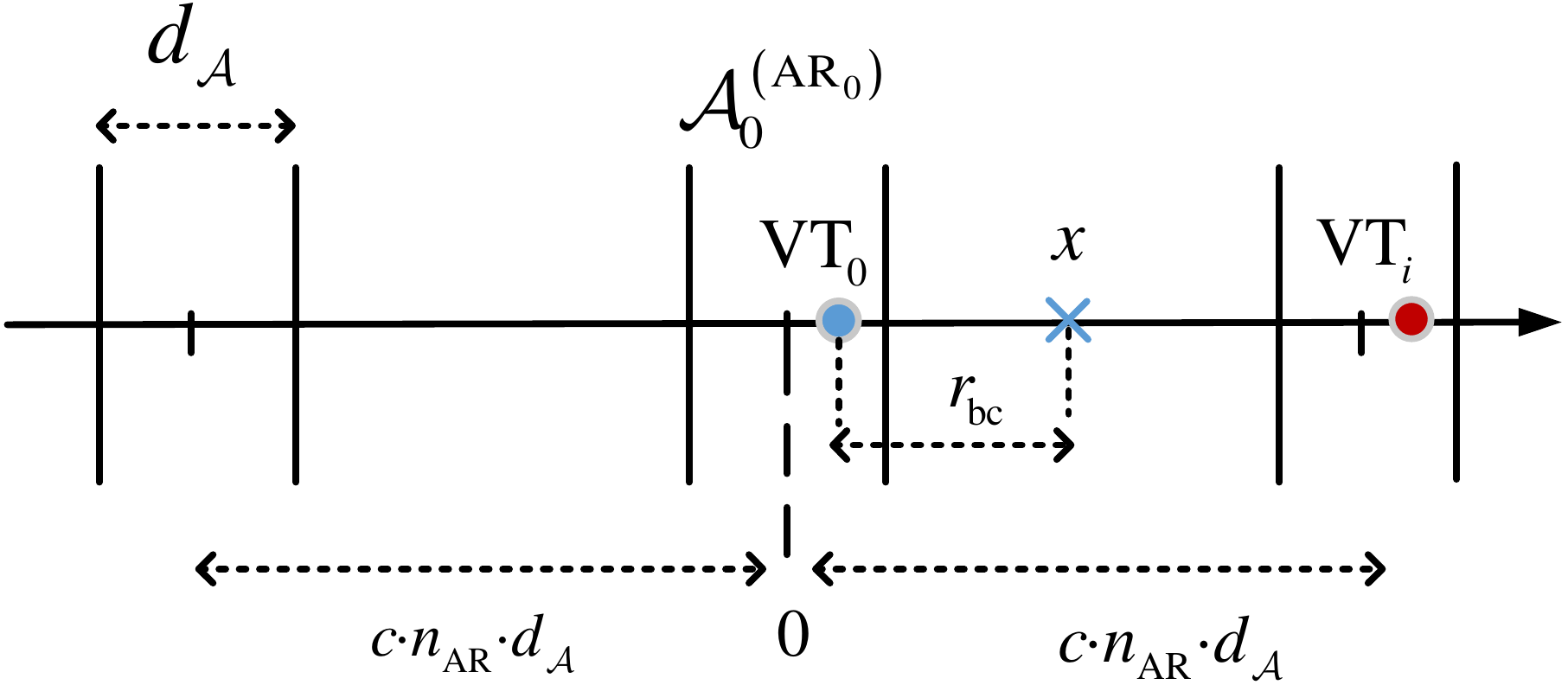}
\caption{Illustration of co-channel segments in the abstraction model. The probe transmitter is represented as $\mathrm{VT}_0$, the probe AR and segment as $\mathrm{AR}_0$ and $\mathcal{A}_{0}^{(\mathrm{AR}_0)}$ respectively, the probe receiver as a blue $x$ and a single interfering vehicle as $\mathrm{VT}_i$. The length of each segment is represented as $d_\mathcal{A}$ and thus $n_\mathrm{AR} \cdot d_\mathcal{A}$ is the minimum co-channel distance.}
\label{fig:abstractionModel}
\end{figure}

\subsection{Signal Modeling}
Transmitted signals suffer from Rayleigh fading, hence the channel power gain is exponentially distributed with unitary mean. Path loss is modeled through a path loss slope $\tau$ and a path loss exponent $\alpha$. Having a receiver placed at location $x$, the SINR can be expressed as follows

\begin{equation}
\label{eq:SINR}
\mathrm{SINR}\left( x \right)=\frac{H_{\mathrm{VT}_{0}}\left( \tau |\mathrm{VT}_{0}-x| \right)^{-\alpha }\rho_{\mathrm{VT}}}{I\left( x \right)+\sigma _{n}^{2}}
\end{equation}

\noindent where $|\cdot|$ the Euclidean distance, $H_{\mathrm{VT}_{0}}$ is the fading power gain between $\mathrm{VT}_0$ and the test receiver, $\rho_\mathrm{VT}$ is the transmit power per Hz, $\sigma_n^2$ is the noise power  and $I(x)$ the received interference at $x$. It is assumed that there is a maximum communication range given by $d_\mathrm{max}$, hence a receiver cannot detect signals from transmitters farther than $d_\mathrm{max}$. 
Thus the interference term can be expressed as follows

\begin{align}
\label{eq:I(x)}
  & I\left( x \right)=\sum\limits_{\mathrm{VT}_{i}\in \Phi ^{(\mathrm{a})}\setminus \left\{ \mathrm{VT}_{0} \right\}}{H_{\mathrm{VT}_{i}}\left( \tau |\mathrm{VT}_{i}-x| \right)^{-\alpha }} \nonumber \\ 
 & \times \rho_{\mathrm{VT}}\mathbf{1}\left( \mathrm{VT}_{i}\in \mathcal{A}^{(\mathrm{AR}_{0})} \right) \mathbf{1}\left( \mathrm{VT}_{i}\in \mathfrak{b}_x \left( d_{\max } \right) \right)
\end{align}

\noindent being $\Phi^{(a)}$ the set of active vehicles in the abstraction model; $\mathbf{1} (\cdot)$ the indicator function and 
$\mathfrak{b}_x (d_\mathrm{max})$ the ball centered at $x$ with radius $d_\mathrm{max}$. 
$H_{\mathrm{VT}_{i}}$ is the fading between $\mathrm{VT}_{i}$ and $x$. 
As it has been justified before, in case of SLP, $\Phi^{(a)}$ is a PPP obtained through independent thinning, with probability $p_a$, from $\Phi$. However, in case of MLP, $\Phi^{(a)}$ is obtained through independent thinning over $\Phi$, which is now a Mat\'ern HCPP of type II. 
Such point process is difficult to analyze because its probability generating functionals do not exist, \cite{Haenggi11}. Hence, the following assumption over the interference is proposed for the sake of tractability. 
\begin{assumption}
\label{ass:HCPP}
The interference term in case of MLP can be approximated as appears below
\begin{align}
  & I\left( x \right)=\sum\limits_{\mathrm{VT}_{i}\in \Phi ^{(\mathrm{a})}}{H_{\mathrm{VT}_{i}}\left( \tau |\mathrm{VT}_{i}-x| \right)^{-\alpha } \rho_{\mathrm{VT}}} \mathbf{1}\left( \mathrm{VT}_{i}\in \mathcal{A}^{(\mathrm{RB}_{0})} \right)
 \nonumber \\ & \quad
 \times \mathbf{1}\left( \left| \mathrm{VT}_{i}-x \right|>d_{\rm safe } \right)\mathbf{1}\left( \left| \mathrm{VT}_{i}-\mathrm{VT}_{0} \right|>d_{\rm safe } \right) 
\end{align}
\noindent where $\Phi^{(a)}$ is a PPP with density $\lambda_L p_a$. 
\end{assumption}

The reasoning behind \textbf{Assumption \ref{ass:HCPP}} is explained below. For tractability, it is assumed that $\Phi^{(a)}$ follows a PPP, instead of a thinned version of a HCPP that represents the locations of active vehicles within the same lane. 
The correlations in the actual point process are captured in the form of spatial constraints by means of a dependent thinning with two indicator functions. 
These constrains guarantee that there are no vehicles nearer than $d_{\rm safe}$ to the probe transmitter nor to the test receiver, which is placed at $x$. It should be noticed that this dependent thinning leads to a point process which is not a homogeneous PPP. 

\begin{remark}[Exact analysis]
\label{rem:Exact analysis}
In the forthcoming analysis: 
(i) the results for the SLP case are exact, since there is PPP as the generative process of the locations of active vehicles; and (ii) the results for the MLP case are approximations due to \textbf{Assumption \ref{ass:HCPP}}. 
\end{remark}

\subsection{Key Performance Indicators}
The capture probability represents the probability that a message is correctly received. Having a receiver placed at location $x$, it is expressed as the probability of the SINR being higher than a given threshold, $\gamma$, which is the CCDF of the SINR. 

We consider that each transmitter uses a fixed Modulation and Coding Scheme (MCS). Here, we use the same abtraction as in \cite{Cheung12, Lin15} to represent the goodput, or equivalently the BR of correctly received bits as

\begin{equation}
\label{eq:BR}
\mathrm{BR}\left( x \right)=\mathbf{1}\left( \mathrm{SINR}\left( x \right)>\gamma  \right)\cdot b_{\rm AR}
\cdot \log _{2}\left( 1+\gamma  \right)
\end{equation}

\noindent where $b_{\rm AR}$ represents the bandwidth associated with a given AR and it is given in Table \ref{tab:SLP and MLP} for SLP and MLP. On the other hand, $\gamma$, which is a system-level parameter, represents the SINR threshold associated with the considered MCS, where $\log_2 (1+\gamma)$ is its spectral efficiency in terms of bps/Hz. 

The EE is defined as the quotient between the BR and the transmit power in a given AR, which can be written in terms of b/J as follows 

\begin{equation}
\label{eq:EE}
\mathrm{EE} \left( x \right)=\frac{\log _{2}\left( 1+\gamma  \right)}{\rho _{\mathrm{VT}}}\mathbf{1}\left( \mathrm{SINR}(x)>\gamma  \right)
\end{equation}

\noindent where the bandwidth term, $b_w$, is canceled out since it appears in the definition of BR and also in the expression of the transmit power in a given AR. 

\subsection{Broadcast Message Types}
As it is mentioned in the introduction, there are two types of broadcast messages: non-periodic and periodic. For non-periodic messages, it is assumed that the probability of being active, i.e., with data to transmit, depends on traffic conditions and other related human issues and thus it is a fixed parameter. 
The case of periodic messages is different, since in this case, the probability of being active depends on the periodic rate and the time needed to transmit the message. This latter metric depends on the spectral efficiency of the MCS and also of the AR bandwidth, $b_{\rm AR}$. Hence, for the case of periodic messages, the probability of being active is expressed as follows

\begin{equation}
\label{eq:pa periodic}
p_{a}=\frac{m_{\mathrm{bc}}}{b_{\mathrm{AR}}t_{\mathrm{rep}}\log _{2}\left( 1+\gamma  \right)}
\end{equation}
\noindent where $t_{\mathrm{rep}}$ is the reporting latency, i.e., the time between two consecutive messages and $m_{\rm bc}$ is the message size in bits. It should be noticed that the time required to transmit the message, which is $m_{\rm bc}/\left( b_{\rm AR} \log_2 (1+\gamma) \right)$, cannot be higher than $t_{\rm rep}$. This imposes the following inequality over the above parameters that must be fulfilled
%
$m_{\mathrm{bc}}/\left( b_{\mathrm{AR}}\log _{2}\left( 1+\gamma  \right) \right)<t_{\mathrm{rep}}$.

\begin{remark}[Density of active vehicles]
\label{rem:Density of active vehicles}
In view of (\ref{eq:pa periodic}) and Table \ref{tab:SLP and MLP} it should be noticed that the density of active vehicles transmitting periodic messages, $\lambda \cdot p_a$, is the same for SLP and MLP schemes for the same $n_{\rm AR}$. 
\end{remark}

\section{Analysis of the Interference and Capture Probability}
\label{sec:Performance Analysis} 

In this section, main performance metrics related to Single-Lane Partition and Multi-Lane Partition are derived. 
Given the broadcast nature of the considered transmissions, a probe receiver placed at a distance 
$r_\mathrm{bc}$ from the probe transmitter, $\mathrm{VT_0}$, is considered. 
Hence, the metrics of interest - capture probability, average BR and average EE - are evaluated at $x=\mathrm{VT}_0 + r_\mathrm{bc}$.

\subsection{Single-Lane Partition (SLP)}
\label{sec:Analysis SLP}

\begin{figure*}[t]
\normalsize 
\begin{align}
\label{eq:kappa}
& \kappa \left( c,s,x \right)=\sum\limits_{j\in \{1,2\}}\mathbf{1}\left( \mu _{L}^{(j)}<\mu _{U}^{(j)} \right)
\Bigg(\mu _{U}^{(j)} {_2F_1}\left( 1,\frac{1}{\alpha },1+\frac{1}{\alpha },\frac{\left( \tau \mu _{U}^{(j)} \right)^{\alpha }}{-s\cdot \rho _{\mathrm{VT}}} \right)
-\mu _{L}^{(j)} {_2F_{1}}\left( 1,\frac{1}{\alpha },1+\frac{1}{\alpha },\frac{\left( \tau \mu _{L}^{(j)} \right)^{\alpha }}{-s\cdot \rho _{\mathrm{VT}}} \right) \Bigg)
\end{align}
\hrulefill
\end{figure*}

To obtain the capture probability, we first compute the Laplace transform the interference, which is given with the following lemma.

\begin{lemma}
\label{lem:SLP LTI}
With SLP, the Laplace transform of the interference evaluated at $x\in \mathbb{R}$ can be written as 
\begin{align}
\label{eq:LTI}
\mathcal{L}_{I\left( x \right)}\left( s \right)=\exp \left(-\lambda \cdot p_{a}\sum\limits_{c=-\left\lfloor d_{\max }/\left( n_{\mathrm{AR}}d_{\mathcal{A}} \right) \right\rfloor }^{\left\lceil d_{\max }/\left( n_{\mathrm{AR}}d_{\mathcal{A}} \right) \right\rceil }{\kappa \left( c,s,x \right)} \right)
\end{align}

\noindent where 
the function $\kappa\left( c,s,x \right)$\footnote{The dependence of functions $\mu _{L}^{(j)}\left( c,x \right)$ and $\mu _{U}^{(j)}\left( c,x \right)$ with $c$ and $x$ has not been written in (\ref{eq:kappa}), (\ref{eq:muL muU}) and (\ref{eq: LTI_1}) for convenience.} appears in (\ref{eq:kappa})
 and 

\begin{align}
\label{eq:muL muU}
& \mu _{U}^{(1)}=\min \left( c\cdot n_{\mathrm{AR}}\cdot d_{\mathcal{A}}+\frac{d_{\mathcal{A}}}{2}-x,d_{\max } \right) \nonumber \\ 
& \mu _{L}^{(2)}=-\max \left( c\cdot n_{\mathrm{AR}}\cdot d_{\mathcal{A}}-\frac{d_{\mathcal{A}}}{2}-x,-d_{\max } \right) 
\nonumber  \\ &
\mu _{L}^{(1)}=\max \left( \mu _{L}^{(2)},0 \right)
; \; \mu _{U}^{(2)}=-\min \left( \mu^{(1)}_{U},0 \right)
\end{align}
\end{lemma}

\begin{proof}
The proof is given in Appendix \ref{app:Proof of lem SLP LTI}.
\end{proof}

\begin{theorem}
\label{theorem:CCDF SINR}
With SLP, the CCDF of the SINR, or equivalently the capture probability, at a distance $r_\mathrm{bc}$ from the typical vehicle transmitter, $\mathrm{VT}_0$, appears below

\begin{align}
\label{eq:CCDF SINR}
& \bar{F}_{\mathrm{SINR}\left( {\rm VT_0} + r_{\rm bc} \right)}\left( \gamma  \right)=\frac{{\mathrm{e}^{-\frac{\gamma \sigma_{n}^{2}}{\rho_{\mathrm{VT}}}  \left( \tau r_{\mathrm{bc}} \right)^{\alpha }}} }{d_{\mathcal{A}}}
\nonumber \\ & \quad \times 
\int\limits_{v=-\frac{d_{\mathcal{A}}}{2}}^{\frac{d_{\mathcal{A}}}{2}} \mathcal{L}_{I\left( {v + r_{\rm bc}} \right)}\left( \frac{\gamma \left( \tau r_{\mathrm{bc}} \right)^{\alpha }}{\rho_{\mathrm{VT}}} \right) \mathrm{d}v 
\end{align}

\noindent where $\mathcal{L}_{I\left( x \right)}\left( s\right)$ is the Laplace transform of the interference, which is given in \textbf{Lemma \ref{lem:SLP LTI}}, with $s = \frac{\gamma }{\rho_{\mathrm{VT}}}\left( \tau r_{\mathrm{bc}} \right)^{\alpha }$ and 
$x = \mathrm{VT}_0 + r_\mathrm{bc}$.
\end{theorem}

\begin{proof}
Since the probe transmitter is chosen at random from the set of active vehicles, its position inside the probe cluster, which is represented as $\mathrm{VT}_0$, is uniformly distributed within the interval $[-d_{\mathcal{A}}/2,d_{\mathcal{A}}/2)$. Hence, the CCDF of the SINR at 
$x=\mathrm{VT}_{0}+r_{\mathrm{bc}}$ can be written as

\begin{align}
& \bar{F}_{\mathrm{SINR}\left( \mathrm{VT}_0 + r_\mathrm{bc} \right)}\left( \gamma  \right)=\Pr \left( \mathrm{SINR}\left( \mathrm{VT}_0 + r_\mathrm{bc} \right)>\gamma  \right) 
\nonumber \\ & \, 
\overset{\mathrm{(a)}}{=} \mathbb{E}_{\mathrm{VT}_{0}}\left[ \Pr \left( H_{\mathrm{VT}_{0}}>
\frac{\gamma \left( \tau r_{\rm bc} \right)^{\alpha } }{\rho_{\mathrm{VT}}}\left( I\left( \mathrm{VT}_0 + r_\mathrm{bc} \right)+\sigma _{n}^{2} \right)
 \right) \right] 
\nonumber \\ & \, 
\overset{\mathrm{(b)}}{=} \mathbb{E}_{\mathrm{VT}_{0}}\mathbb{E}_{I}\left[ \mathrm{e}^{-\frac{\gamma }{\rho_{\mathrm{VT}}}\left( I\left(  \mathrm{VT}_{0}+r_{\mathrm{bc}} \right) +\sigma _{n}^{2} \right)\left( \tau r_{\mathrm{bc}} \right)^{\alpha }} \right] 
\nonumber \\ & \, 
\overset{\mathrm{(c)}}{=} 
\frac{\mathrm{e}^{-\frac{\gamma }{\rho_{\mathrm{VT}}}\sigma _{n}^{2}\left( \tau r_{\mathrm{bc}} \right)^{\alpha }}}{d_{\mathcal{A}}}\int\limits_{v=-d_{\mathcal{A}}/2}^{d_{\mathcal{A}}/2}\mathcal{L}_{I\left( v+r_{\rm bc} \right)}\left( \frac{\gamma }{\rho_{\mathrm{VT}}}\left( \tau r_{\mathrm{bc}} \right)^{\alpha } \right)\cdot \mathrm{d}v 
\end{align}

\noindent where (a) comes after reordering the expression of the SINR and applying the total probability theorem over position $\mathrm{VT}_0$; (b) after performing expectation over the fading and conditioning over the interference term and (c) after expressing the expectation over ${\rm VT}_0$ in integral form and identifying the Laplace transform of the interference. 
\end{proof}


\begin{corollary}
\label{cor:nRB to inf}
The capture probability with SLP in the limiting case where $n_{\rm AR} \to \infty$ is given as follows
\begin{align}
& \lim\limits_{n_\mathrm{AR}\to \infty} \bar{F}_{\mathrm{SINR}({\rm VT_0}+r_{\rm bc})}\left( \gamma  \right)=\frac{e^{-\frac{\gamma \sigma _{n}^{2}}{\rho _{\mathrm{VT}}}\left( \tau r_{\mathrm{bc}} \right)^{\alpha }}}{d_{\mathcal{A}}}
\nonumber \\ & \quad
\int\limits_{v=-\frac{d_{\mathcal{A}}}{2}}^{\frac{d_{\mathcal{A}}}{2}}{\exp \left( -\lambda p_{a}\kappa \left( 0,\frac{\gamma \sigma _{n}^{2}}{\rho _{\mathrm{VT}}}\left( \tau r_{\mathrm{bc}} \right)^{\alpha },v+r_{\mathrm{bc}} \right) \right)} \mathrm{d}v
\end{align}
\end{corollary}

\begin{proof}
The proof follows the fact that when $n_\mathrm{AR}\to\infty$ the indicator function given in (\ref{eq: LTI_1}), $\mathbf{1}\left( y\in \mathfrak{b}_{x}\left( d_{\max } \right) \right)$, is non zero only for $c=0$.  
\end{proof}

\begin{remark}[Intra-segment interference limited regime]
\label{rem:Intra-segment interference}
In view of Corollary \ref{cor:nRB to inf} it can be observed that the capture probability when $n_\mathrm{AR}$ tends to infinity is limited by the interference of the probe segment ($c=0$), which is related to those cases where an interfering vehicle is transmitting in the same segment as the probe vehicle transmitter. 
\end{remark}

It has been necessary to obtain the Laplace transform of the interference to compute the CCDF of the SINR. Besides, the Laplace transform of the interference is useful to obtain the average interference, which provides further insights. The following Lemma gives such result.

\begin{lemma}
\label{lemma:avI}
The average received interference at $x$, being the probe segment centered at the origin can be expressed as 

\begin{align}
\label{eq:avI}
& \mathbb{E}\left[ I\left( x \right) \right]=\frac{\lambda p_{a}\rho _{\mathrm{VT}}}{\alpha -1}\sum\limits_{c=-\left\lfloor d_{\max }/\left( n_{\mathrm{AR}}d_{\mathcal{A}} \right) \right\rfloor }^{\left\lceil d_{\max }/\left( n_{\mathrm{AR}}d_{\mathcal{A}} \right) \right\rceil }\sum\limits_{j\in \{1,2\}}
\nonumber \\ & \quad
\left( \mu _{L}^{(j)}\left( c,x \right)\left( \mu _{U}^{(j)}\left( c,x \right) \right)^{\alpha } \right.-\tau ^{-\alpha }\left( \mu _{U}^{(j)}\left( c,x \right) \right)^{1-\alpha }
\nonumber \\ & \quad
\times {\mathbf{1}\left( \mu _{L}^{(j)}\left( c,x \right)<\mu _{U}^{(j)}\left( c,x \right) \right)}
\end{align}

\end{lemma}

\begin{proof}
Using the fact that the Laplace transform can be used as a moment generating function, the average interference can be written as 
$\mathbb{E}\left[ I (x) \right]=-\left| \frac{\mathrm{d}}{\mathrm{d}s}\mathcal{L}_{I(x)}\left( s \right) \right|_{s=0}$. Hence the proof consists on obtaining the derivative of (\ref{eq:LTI}) and then particularizing for $s=0$. 

\end{proof}

\begin{remark}[Convergence of the interference]
\label{rem:Covergence of the interference}
In view of (\ref{eq:avI}) it can be stated that the average interference is only finite for $x \notin \mathfrak{b}_{c \cdot n_\mathrm{AR} \cdot d_\mathcal{A}} (d_\mathcal{A}/2)$, since for $x \in \mathfrak{b}_{c \cdot n_\mathrm{AR} \cdot d_\mathcal{A}}(d_\mathcal{A}/2)$ we have $\mu^{(2)}_U (c,x)=0$ which makes   the average interference tends to infinity. 
\end{remark}

%

\subsection{Multi-Lane Partition (MLP)}
\label{sec:Analysis MLP}

\begin{figure*}[t]
\normalsize 
\begin{align}
\label{eq:zeta}
& \zeta \left( c,s,x \right)=\sum\limits_{i\in \{1,2\}} \sum\limits_{j\in \{1,2\}}\mathbf{1}\left( \mu _{L}^{(i,j)}<\mu _{U}^{(i,j)} \right)
\Bigg(\mu _{U}^{(i,j)} {_2F_1}\left( 1,\frac{1}{\alpha },1+\frac{1}{\alpha },\frac{\left( \tau \mu _{U}^{(i,j)} \right)^{\alpha }}{-s\cdot \rho _{\mathrm{VT}}} \right)
-\mu _{L}^{(i,j)} {_2F_{1}}\left( 1,\frac{1}{\alpha },1+\frac{1}{\alpha },\frac{\left( \tau \mu _{L}^{(i,j)} \right)^{\alpha }}{-s\cdot \rho _{\mathrm{VT}}} \right) \Bigg)
\end{align}
\hrulefill
\end{figure*}

The Laplace transform of the interference for the case of MLP is given by the following lemma. 

\begin{lemma}
\label{lem:MLP LTI}
In case of MLP, the Laplace transform of the interference evaluated at the probe receiver, placed at $x$ is given by
\begin{equation}
\mathcal{L}_{I\left( x \right)}\left( s \right)=\exp \left( -\lambda \cdot p_{a}\sum\limits_{c=-\left\lfloor d_{\max }/\left( n_{\mathrm{RB}}d_{\mathcal{A}} \right) \right\rfloor }^{\left\lceil d_{\max }/\left( n_{\mathrm{RB}}d_{\mathcal{A}} \right) \right\rceil }{\zeta \left( c,s,x \right)} \right)
\end{equation}
\noindent where 
$\zeta \left( c,s,x \right)$ is written in (\ref{eq:zeta}) and

\begin{align}
\label{eq:muijL muijU 1}
& \mu _{U}^{(1,1)}=\min \left( c\cdot n_{\mathrm{RB}}\cdot d_{\mathcal{A}}+\frac{d_{\mathcal{A}}}{2}-x,d_{\max } \right) 
\nonumber \\ 
& \mu _{L}^{(2,2)}=-\max \left( c\cdot n_{\mathrm{RB}}\cdot d_{\mathcal{A}}-\frac{d_{\mathcal{A}}}{2}-x,-d_{\max } \right) 
\nonumber \\ 
& \mu _{L}^{(1,1)}=\max \left( \mu _{L}^{(2,2)},d_{\mathrm{safe} },d_{\mathrm{safe} }-x+v \right) 
\nonumber \\ 
& \mu _{U}^{(2,1)}=\min \left( \mu _{U}^{(1,1)},v-x-d_{\mathrm{safe} } \right) 
\nonumber \\ 
& \mu _{L}^{(2,1)}=\max \left( \mu _{L}^{(2,2)},d_{\mathrm{safe} } \right) 
\nonumber \\ 
& \mu _{L}^{(2,1)}=\max \left( \mu _{L}^{(2,2)},d_{\mathrm{safe} } \right) 
\nonumber \\ 
& \mu _{U}^{(1,2)}=-\min \left( \mu _{U}^{(1,1)},-d_{\mathrm{safe} } \right) 
\nonumber \\ 
& \mu _{L}^{(1,2)}=-\max \left( \mu _{L}^{(2,2)},d_{\mathrm{safe} }-x+v \right) 
\nonumber \\ 
& \mu _{U}^{(2,2)}=-\min \left( \mu _{U}^{(1,1)},-d_{\mathrm{safe} },v-x-d_{\mathrm{safe} } \right) 
\end{align}
\end{lemma}

\begin{proof}
The proof is given in Appendix \ref{app:Proof of lem MLP LTI}. 
\end{proof}

\begin{corollary}
In the special case, $d_{\mathcal{A}} < d_{\rm safe} < (n_{\rm AR} - 1) d_{\cal A}$ and $|x| < n_{\rm AR} d_{\cal A}/2$, the Laplace transform of the interference can be simplified into the following expression
\begin{equation}
\mathcal{L}_{I\left( x \right)}\left( s \right)=\exp \left( -\lambda \cdot p_{a}\sum\limits_{c=-\left\lfloor \frac{d_{\max }}{ n_{\mathrm{AR}}d_{\mathcal{A}} } \right\rfloor }^{\left \lceil \frac{d_{\max }}{ n_{\mathrm{RB}}d_{\mathcal{A}} } \right\rceil }{\kappa \left( c,s,x \right)}\mathbf{1}\left( c\ne 0 \right) \right)
\end{equation}
\noindent where $\lambda=\lambda_L$ for the MLP case (Table \ref{tab:SLP and MLP}). 
\end{corollary}

\begin{proof}
The proof consists on noticing that in the case where $d_{\mathcal{A}} < d_{\rm safe}< (n_{\rm AR} - 1) d_{\cal A}$ holds, then, the indicator function $\mathbf{1}(|{\rm VT_i}-{\rm VT_0}|>d_{\rm safe})$ is equal to $0$ if ${\rm VT_i} \in [-d_{\cal A}/2,d_{\cal A}/2)$ and $1$ otherwise. This means that there is no interfering vehicles inside the probe segment. 
Additionally, 
if $|x| < n_{\rm AR} d_{\cal A}/2$, then $\mathbf{1}(|{\rm VT_i}-x|>d_{\rm safe})=1$. 
Hence, in view of \textbf{Assumption \ref{ass:HCPP}}, the analysis is analogous to the case of SLP, but taking into account that there is no intra-segment interference, which is captured in the indicator function $\mathbf{1}(c \neq 0)$. 
\end{proof}

The next theorem gives the capture probability with MLP. 

\begin{theorem}
\label{theorem:MLP CCDF SINR}
The capture probability of a probe receiver placed at a distance $r_{\rm bc}$ from the transmitter with MLP is

\begin{align}
& \bar{F}_{\mathrm{SINR}\left( \mathrm{VT}_{0}+r_{\mathrm{bc}} \right)}\left( \gamma  \right)=
\mathrm{e}^{-\frac{\gamma }{p_{\mathrm{VT}}}\sigma _{n}^{2}\left( \tau r_{\mathrm{bc}} \right)^{\alpha }}
\int\limits_{v=-d_{\mathcal{A}}/2}^{d_{\mathcal{A}}/2}
\nonumber \\ & \quad
\times \frac{\mathbf{1}\left( v\notin \mathfrak{b}_{v+r_{\mathrm{bc}}}\left(d_{\mathrm{safe} } \right) \right)}{\left| \mathcal{D}\left( v+r_{\mathrm{bc}} \right) \right|} \mathcal{L}_{I\left( v+r_{\mathrm{bc}} \right)}\left( \frac{\gamma }{p_{\mathrm{VT}}}\left( \tau r_{\mathrm{bc}} \right)^{\alpha } \right)\cdot \mathrm{d}v
\end{align}
\noindent where $\mathcal{L}_{I(x)}(s)$ is given in \textbf{Lemma \ref{lem:MLP LTI}} and $\left| \mathcal{D}\left( x \right) \right|$, which represents the Lebesgue measure of the relative complement of the interval $[-d_{\cal A}/2,d_{\cal A}/2)$ with respect to the set 
$\mathfrak{b}_{x}\left( d_{\mathrm{safe} } \right)$, and it is written as 

\begin{equation}
\left| \mathcal{D}\left( x \right) \right|=\int\limits_{w=-d_{\mathcal{A}}/2}^{d_{\mathcal{A}}/2}{\mathbf{1}\left( v\notin \mathfrak{b}_{x}\left( d_{\mathrm{safe}} \right) \right)}\cdot \mathrm{d}w
\end{equation}

\end{theorem}

\begin{proof}
The proof follows from having the probe vehicle uniformly distributed inside the region 
$\mathcal{D}\left( x \right)=[-d_{\mathcal{A}}/2,d_{\mathcal{A}}/2)
\setminus \mathfrak{b}_x \left( d_{\rm safe} \right)$, and hence the pdf of its position is given as 
$f_{\mathrm{VT}_{0}}\left( v \right)=\mathbf{1}\left( v\in \mathcal{D}\left( x \right) \right)/\left| \mathcal{D}\left( x \right) \right|$. Then, conditioning over the position of the probe vehicle and over interference, and reordering completes the proof.
\end{proof}

\begin{remark}[Exponential dependence]
\label{rem: MLP Dependence with noise and transmit power}
In view of \textbf{Theorems \ref{theorem:CCDF SINR}}  and \textbf{\ref{theorem:MLP CCDF SINR}}, it can be observed that the capture probability for both, SLP and MLP, only depends on $\rho_{\rm VT}$ as $c^{(k)}_1 \exp(-c_2/\rho_{\rm VT})$, with the label $k$ being either equal to SLP or MLP, i.e., $k=\{{\rm SLP,MLP}\}$. Such an expression is an increasing function with respect to $\rho_{\rm VT}$, where $c^{(k)}_1$ and $c_2$ depend on other system parameters, and thus, they are constants with respect to $\rho_{\rm VT}$.  Analogously, the capture probability depends on the noise power, $\sigma^2_n$, as $c^{(k)}_1 \exp(-c^{\triangle}_2 \sigma^2_n)$, which is a decreasing function with respect to $\sigma^2_n$. 
Therefore, the maximum capture probability, for a given set of system parameters, is equal to $c_1^{(k)}$, and it is achieved either in the limit, $ \rho_{\rm VT} \to \infty$, or in the no-noise case 
($\sigma_n^2 = 0$).
\end{remark}

\begin{remark}[Noise-limited regime]
\label{rem:Noise-limited regime}
The system is noise limited, and thus, there is no interference if $d_{\rm safe} > d_{\cal A}$ and 
$n_{\rm AR} > \left(2 d_{\rm max} + d_{\cal A} \right)/d_{\cal A}$. 
\end{remark}

\begin{proof}
With MLP, if $d_{\rm safe} > d_{\cal A}$ there is no intra-segment interference. Hence, in this case it is possible to determine the number of ARs, $n_{\rm AR}$, that leads to a system without interference. This is guaranteed if the distance between the probe receiver, which is placed at 
$x={\rm VT_0}+r_{\rm bc}$, and the nearest interfering vehicle in the nearest co-channel segment is higher than $d_{\rm max}$ (worst case scenario). Such an scenario involves that the probe vehicle is placed at ${\rm VT_0}=d_{\cal A}/2$ and the probe receiver is placed at the maximum communication range, with $r_{\rm bc}=d_{\rm max}$. Therefore, in this case, the nearest interfering vehicle must be placed at a distance towards the probe receiver greater than $d_{\rm max}$. This requires that $n_{\rm AR} d_{\cal A}- d_{\cal A}> 2 d_{\rm max}$. Reordering the above inequality completes the proof.
\end{proof}

As it can be noticed from \textbf{Remark \ref{rem:Noise-limited regime}}, by augmenting $n_{AR}$ it is possible to assure no interference, which allows to greatly increase the capture probability  by increasing the transmit power. 

\begin{corollary}
In the special case of $r_{\rm bc} > d_{\cal A}$, the capture probability is given as 

\begin{align}
& \bar{F}_{\mathrm{SINR}\left( \mathrm{VT}_{0}+r_{\mathrm{bc}} \right)}\left( \gamma  \right)=\frac{\mathrm{e}^{-\frac{\gamma }{p_{\mathrm{VT}}}\sigma _{n}^{2}\left( \tau r_{\mathrm{bc}} \right)^{\alpha }}}{d_{\mathcal{A}}}\int\limits_{v=-d_{\mathcal{A}}/2}^{d_{\mathcal{A}}/2}
\nonumber \\ & \quad
{\mathcal{L}_{I\left( v+r_{\mathrm{bc}} \right)}\left( \frac{\gamma }{p_{\mathrm{VT}}}\left( \tau r_{\mathrm{bc}} \right)^{\alpha } \right)}\cdot \mathrm{d}v
\end{align}
\noindent where the Laplace transform of the interference is now given in \textbf{Lemma \ref{lem:MLP LTI}}. 
\end{corollary}

\begin{proof}
The proof comes after realizing that, in case of $r_{\rm bc} > d_{\cal A}$, then $|\mathcal{D}(v+r_{\rm bc})|$ is $d_{\cal A}$ and $\mathbf{1}\left( v\notin \mathfrak{b}_{v+r_{\rm bc}} \left( d_{{\rm safe}} \right) \right)$=1. 
\end{proof}

\begin{table*}
\renewcommand{\arraystretch}{1.3}
\caption{Summary of mathematical results as functions of $\rho_{\rm VT}$}
\label{tab:Summary Results}
\centering
\begin{tabular}{ |c| c| c| c| }
\hline 
Metric & 
$\bar{F}_{\mathrm{SINR}}(\gamma)$
& $\mathbb{E}\left[ \mathrm{BR} \right]$
& $\mathbb{E}\left[ \mathrm{EE} \right]$ \\
\hline 
\hline 
SLP & $c_{1}^{\mathrm{(SLP)}}\mathrm{e}^{-\frac{c_{2}}{\rho _{\mathrm{VT}}}}$ 
& 
$\frac{b_{w}}{n_{\mathrm{RB}}}\cdot \log _{2}\left( 1+\gamma  \right)\cdot c_{1}^{(\mathrm{SLP})}\mathrm{e}^{-\frac{c_{2}}{\rho _{\mathrm{VT}}}}$ 
& $\frac{\log _{2}\left( 1+\gamma  \right)}{\rho _{\mathrm{VT}}}\cdot c_{1}^{(\mathrm{SLP})}\mathrm{e}^{-\frac{c_{2}}{\rho _{\mathrm{VT}}}}$ \\
\hline
MLP & $c_{1}^{\mathrm{(MLP)}}\mathrm{e}^{-\frac{c_{2}}{\rho _{\mathrm{VT}}}}$ 
& 
$\frac{b_{w}}{n_{\mathrm{RB}}\cdot n_L}\cdot \log _{2}\left( 1+\gamma  \right)\cdot c_{1}^{(\mathrm{MLP})}\mathrm{e}^{-\frac{c_{2}}{\rho _{\mathrm{VT}}}}$ 
& $\frac{\log _{2}\left( 1+\gamma  \right)}{\rho _{\mathrm{VT}}}\cdot c_{1}^{(\mathrm{MLP})}\mathrm{e}^{-\frac{c_{2}}{\rho _{\mathrm{VT}}}}$ \\
\hline
\multicolumn{2}{|c|}{$
c_{2}=\gamma \sigma _{n}^{2}\left( \tau r_{\mathrm{bc}} \right)^{\alpha }$} & \multicolumn{2}{c|}{$
c_{1}^{(\mathrm{SLP})}=\frac{1}{d_{\mathcal{A}}}\int\limits_{v=-\frac{d_{\mathcal{A}}}{2}}^{\frac{d_{\mathcal{A}}}{2}}{\mathcal{L}_{I(v+r_{\mathrm{bc}})}}\left( \frac{\gamma \left( \tau r_{\mathrm{bc}} \right)^{\alpha }}{\rho _{\mathrm{VT}}} \right)\mathrm{d}v$} \\
\hline
\multicolumn{4}{|c|}{$
c_{1}^{\mathrm{(MLP)}}=\int\limits_{v=-d_{\mathcal{A}}/2}^{d_{\mathcal{A}}/2}{\frac{\mathbf{1}\left( v\notin \mathfrak{b}_{v+r_{\mathrm{bc}}}\left( d_{\rm safe } \right) \right)}{\left| \mathcal{D}\left( v+r_{\mathrm{bc}} \right) \right|}\mathcal{L}_{I\left( v+r_{\mathrm{bc}} \right)}\left( \frac{\gamma }{p_{\mathrm{VT}}}\left( \tau r_{\mathrm{bc}} \right)^{\alpha } \right)}\cdot \mathrm{d}v$} \\
\hline 
\end{tabular}
\end{table*}

The average interference is given in the following lemma. 

\begin{lemma}
\label{lemma:MLP avI}
With MLP, the average received interference at $x$, being the probe segment centered at the origin, can be expressed as 

\begin{align}
\label{eq:MLP avI}
& \mathbb{E}\left[ I\left( x \right) \right]=\frac{\lambda p_{a}\rho _{\mathrm{VT}}}{\alpha -1}\sum\limits_{c=-\left\lfloor \frac{d_{\max }}{n_{\mathrm{AR}}d_{\mathcal{A}}} \right\rfloor }^{\left\lceil \frac{d_{\max }}{ n_{\mathrm{AR}}d_{\mathcal{A}}} \right\rceil } 
\sum\limits_{i\in \{1,2\}} \sum\limits_{j\in \{1,2\}}
\nonumber \\ & \quad
\left( \mu _{L}^{(i,j)}\left( c,x \right)\left( \mu _{U}^{(i,j)}\left( c,x \right) \right)^{\alpha } \right.-\tau ^{-\alpha }\left( \mu _{U}^{(i,j)}\left( c,x \right) \right)^{1-\alpha }
\nonumber \\ & \quad
\times {\mathbf{1}\left( \mu _{L}^{(i,j)}\left( c,x \right)<\mu _{U}^{(i,j)}\left( c,x \right) \right)}
\end{align}
\end{lemma}

\begin{proof}
The proof is analogous to \textbf{Lemma \ref{lemma:avI}}. Hence, it has been obtained the derivative of the Laplace transform of the interference, which is given in \textbf{Lemma \ref{lem:MLP LTI}}, and the resulting expression  has then been particularized for $s=0$. 
\end{proof}

\section{Binary Rate, Energy Efficiency and Optimal Transmit Power}
\subsection{Binary Rate and Energy Efficiency}
\label{sec:Analysis SLP}
Besides the capture probability, another key performance indicator for system design is the average BR.  This result is given in the following Lemma. 

\begin{lemma}
\label{lemma:avBR}
The average BR at a distance $r_\mathrm{bc}$ from the typical vehicle transmitter, $\mathrm{VT}_0$, appears below

\begin{equation}
\label{eq:avBR}
\mathbb{E}\left[ \mathrm{BR}\left(\mathrm{VT}_0 + r_\mathrm{bc}\right)\right]={b_{\mathrm{AR}}}\cdot \log _{2}\left( 1+\gamma \right)\cdot \bar{F}_{\mathrm{SINR}}\left( \gamma  \right)
\end{equation}

\noindent where $b_{\rm AR}$ is given in Table \ref{tab:SLP and MLP} and $\bar{F}_{\mathrm{SINR}}\left( \gamma  \right)$ is either given in \textbf{Theorem \ref{theorem:CCDF SINR}} or \textbf{\ref{theorem:MLP CCDF SINR}} depending on the considered scheme, i.e., SLP or MLP respectively.
\end{lemma}

\begin{proof}
The proof consists on performing expectation over (\ref{eq:BR}) and realizing that $\mathbb{E}\left[ \mathbf{1}\left( \mathrm{SINR} \left(\mathrm{VT}_0 + r_\mathrm{bc}\right) >\gamma  \right) \right]$ is the CCDF of the SINR.

%
\end{proof}

\begin{remark}[Average rate when $n_{\rm  AR}$ tends to infinity]
\label{rem:Average rate when nRB tends to infinity}
In view of \textbf{Lemma \ref{lemma:avBR}} and \textbf{Corollary \ref{cor:nRB to inf}} it can be stated that for a finite SINR threshold, $\gamma$, the average BR tends to $0$ as $n_{\rm  AR}$ tends to infinity.
\end{remark}

\begin{proof}
The proof consists on noting that the CCDF of the SINR is equal or smaller than $1$, hence for a finite $\gamma$ the term $n_{\rm AR}$ in the denominator of (\ref{eq:avBR}) makes the average BR  tend to $0$. 
\end{proof}

\begin{lemma}
\label{lemma:avEE}
The average EE at a distance $r_\mathrm{bc}$ from the typical vehicle transmitter, $\mathrm{VT}_0$, appears below

\begin{equation}
\label{eq:avEE}
\mathbb{E}\left[ \mathrm{EE}\left(\mathrm{VT}_0 + r_\mathrm{bc}\right) \right]=\frac{\log _{2}\left( 1+\gamma  \right)}{\rho _{\mathrm{VT}}}\cdot \bar{F}_{\mathrm{SINR}\left(\mathrm{VT}_0 + r_\mathrm{bc}\right)}\left( \gamma  \right)
\end{equation}

\noindent where $x = \mathrm{VT}_0 + r_\mathrm{bc}$ and $\bar{F}_{\mathrm{SINR}}\left( \gamma  \right)$ is either given in \textbf{Theorem \ref{theorem:CCDF SINR}} or \textbf{\ref{theorem:MLP CCDF SINR}} depending on the considered scheme, i.e., SLP or MLP respectively.
\end{lemma}

\begin{proof}
The proof is analogous to the case of \textbf{Lemma \ref{lemma:avBR}}.
\end{proof}

In view of \textbf{Remark \ref{rem: MLP Dependence with noise and transmit power}}, the capture probability,  average BR and EE can be written as it appears in Table \ref{tab:Summary Results}.

\subsection{Optimal Transmit Power}
\label{sec:optimization}

In this section, the optimal transmit power that maximizes the EE, subject to a minimum capture probability, is derived. Such a constrain is expressed as a percentage, $\delta$, of the maximum capture probability that can be achieved according to \textbf{Remark \ref{rem: MLP Dependence with noise and transmit power}}. More formally, the optimization problem is formulated as follows

\begin{align}
\label{eq:Optimization problem}
& \underset{\rho _{\mathrm{VT}}}{\mathop{\rm maximize }}\,\mathbb{E}\left[ \mathrm{EE} \right]
\nonumber \\ & 
\text{subject to}\quad \bar{F}_{\mathrm{SINR}}\left( \gamma  \right)
\ge c^{(k)}_{1}\delta 
\nonumber \\ & \quad \quad \quad \quad \quad  
0<\delta <1 
\end{align}

\noindent where $k$ is a label that can be either equal to SLP or MLP, i.e., $k=\{{\rm SLP, MLP}\}$ and the metrics $\mathbb{E}\left[ \mathrm{EE} \right] $ and $\bar{F}_{\mathrm{SINR}}\left( \gamma  \right)$ are given in Table \ref{tab:Summary Results} for SLP and MLP
\footnote{Throughout this section as well as in Table \ref{tab:Summary Results}, it is neglected the dependence with $x = \mathrm{VT}_0 + r_\mathrm{bc}$ in 
$\mathbb{E}\left[ \mathrm{BR}\left( {\rm VT_0}+r_{\rm bc} \right) \right]$, $\mathbb{E}\left[ \mathrm{EE}\left( {\rm VT_0}+r_{\rm bc} \right) \right]$ and $\bar{F}_{{\rm SINR}({\rm VT_0}+r_{\rm bc})}$ for the sake of simplicity.}. 
Solving (\ref{eq:Optimization problem}) leads to the following theorem. 

\begin{theorem}
\label{theo:Optimal transmit power}
The optimal transmit power, for SLP and MLP, can be written as appears below

\begin{equation}
\rho_{\mathrm{VT}}^{\star }=\left\{ \begin{matrix}
   c_{2} & \mathrm{if}\;0<\delta \le \mathrm{e}^{-1}  \\
   c_{2}\ln ^{-1}\left( 1/\delta  \right) & \mathrm{if}\;1>\delta >\mathrm{e}^{-1}  \\
\end{matrix} \right.
\end{equation} 
\noindent leading to the following EE and capture probability

\begin{equation}
\mathbb{E}\left[ \mathrm{EE} \right]^{\star }=\left\{ \begin{matrix}
   \frac{c^{(k)}_{1}}{c_{2}}\log _{2}\left( 1+\gamma  \right)\mathrm{e}^{-1} & \mathrm{if}\;0<\delta \le \mathrm{e}^{-1}  \\
   \frac{c^{(k)}_{1}}{c_{2}}\ln \left( \frac{1}{\delta } \right)\log _{2}\left( 1+\gamma  \right)\delta  & \mathrm{if}\;1>\delta >\mathrm{e}^{-1}  \\
\end{matrix} \right.
\end{equation}

\begin{equation}
\bar{F}_{\mathrm{SINR}}^{\star }\left( \gamma  \right)=\left\{ \begin{matrix}
   c^{(k)}_{1}\mathrm{e}^{-1} & \mathrm{if}\;0<\delta \le \mathrm{e}^{-1}  \\
   c^{(k)}_{1}\delta  & \mathrm{if}\;1>\delta >\mathrm{e}^{-1}  \\
\end{matrix} \right.
\end{equation}

\noindent where, again, $k=\{{\rm SLP, MLP}\}$. 

\end{theorem}

\begin{proof}
The average EE is a concave function in the open interval  $\rho_{\rm VT} \in (0,\infty)$. Hence, it has a single critical point, which is placed at $\rho_{\rm VT}=c_2$, that leads to the global maximum. Nevertheless, the constrain over the capture probability imposes that the solution must lie between the following interval $\rho_{\rm VT} \in [c_2 \ln^{-1} (1/\delta),\infty)$. It should be noticed that the average EE is an increasing function for $\rho_{\rm VT}  < c_2$ and a decreasing function for $\rho_{\rm VT}  > c_2$. Hence, if $\delta < {\rm e}^{-1}$, the critical point fulfills the constrain over the capture probability, which leads to the solution $\rho_{\rm VT}^\star = c_2$. On the other hand, if $\delta > {\rm e}^{-1}$, the constrain governs the optimal transmit power which is now $\rho_{\rm VT}^\star = c_{2}\ln ^{-1}\left( 1/\delta  \right)$.
\end{proof}

\begin{remark}[Independence of the optimal transmit power]
\label{rem:opt}
In view of \textbf{Theorem \ref{theo:Optimal transmit power}} it can be stated that the optimal transmit power is independent of the considered scheme, i.e., SLP or MLP. This is due to the fact that the optimal transmit power only depends on $c_2$ and $\delta$. 
\end{remark}

\section{Numerical Results}
\label{sec:Numerical Results}

\begin{table}
\renewcommand{\arraystretch}{1.3}
\caption{Simulation Parameters}
\label{tab:Simulation Parameters}
\centering
\begin{tabular}{ c c c c }
\toprule
Parameter & Value & Parameter & Value \\
\toprule
$(\tau,\alpha)$ & $(490,1.68)$ & $\lambda_L$ (vehicles/m) & $0.8\cdot 84^{-1}$\\
\hline
$d_{\rm safe}$ (m) & $42$ & $n_L$ & 2 \\
\hline
$d_{\cal A}$ (m) & $42$  & $n_{\rm AR}$ & $10$ \\
\hline
$p_a$ (non-periodic) & $0.25$ & $(m_{\rm bc},t_{\rm rep})$ (bits,ms) & $2400,100$\\
\hline
$\rho_{\rm VT}$ (dBm/Hz) & $-40$ & $\gamma$ (dB) & $5$ \\
\hline
$\sigma_n^2$ (dBm/Hz) & $-165$ & $d_{\rm max}$ (km) & $56$ \\
\hline
$b_w$ (MHz) & $9$ & $r_{\rm bc}$ (m) & 150 \\
\bottomrule
\end{tabular}
\end{table}

In this section, analytical results are illustrated and validated with extensive Monte Carlo simulations in order to assess GLOC performance. 
The simulation setup is chosen from the guidelines given in \cite{Karedal11,3gpp2016, Chen10}. In particular, it is considered a velocity-dependent safe distance between vehicles, according to 
3GPP simulation assumptions for LTE V2X \cite{3gpp2016}: $d_{\rm safe}({\rm m}) = 2.5 \cdot v ({\rm m/s})$. Assuming a vehicle velocity of $60$ km/h, this leads to a safe distance of $42$ m. Such a minimum distance yields to a maximum vehicle density per lane of $\lambda_{L,{\rm max}}=1/(2 d_{\rm safe})=84^{-1}$ \cite{Haenggi13}. A high density of vehicles is considered, and hence the density per lane is set to $80\%$ of the maximum density. The system bandwidth is $10$ MHz, as given in \cite{3gpp2016}; however, excluding guard-bands in LTE this leads to $9$ MHz of useful bandwidth. With non-periodic messages, it is considered a probability of being active, i.e., with data to transmit, of $0.25$. On the other hand, with periodic messages it is considered a message size of $2400$ bits and reporting time of $100$ ms as given in \cite {3gpp2016}.
The path loss is taken from \cite{Karedal11}, where a vast measurement campaign over $5.2$ GHz is performed. 
It is considered that VTs transmit with $-40$ dBm/Hz, a thermal noise power of $-174$ dBm/Hz and a noise figure of $9$ dB as pointed out in \cite{3gpp2016}, which leads to $\sigma_n^2=-165$ dBm/Hz. 
%
%
Simulation parameters are summarized in Table \ref{tab:Simulation Parameters}. Simulations are carried out averaging over $10^4$ spatial realizations. Through this section, analytical results are drawn with solid lines whereas markers are used for simulation results. As stated in \textbf{Remark \ref{rem:Exact analysis}}, results related with SLP are exact whereas results related to MLP are approximations. Nevertheless a good match between simulation and analysis is observed in both cases. 

\subsection{Impact of the SINR threshold and the traffic activity}
\label{sec:Rate and density}

\begin{figure}[t]
\centering
\includegraphics[width=\@figSize]{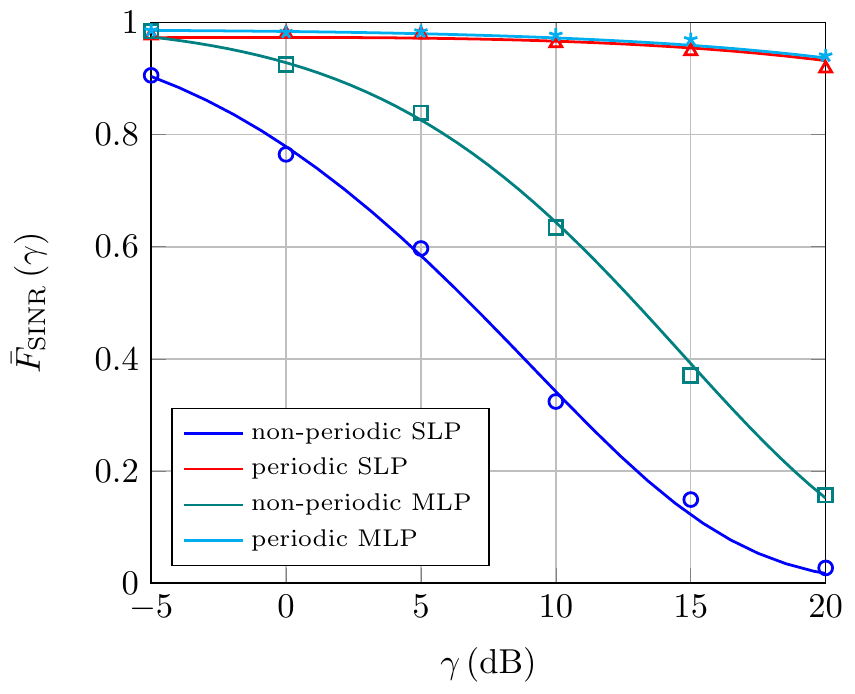}
\caption{Capture probability versus the SINR threshold, $\gamma$, for SLP ($d_{\rm safe}=0$ m) and for MLP ($d_{\rm safe}=42$) with periodic and non-periodic broadcast messages. Analytical results are represented with solid lines whereas simulation results are represented with marks.}
\label{fig:ccdfSINR_gamma} 
\end{figure}

\begin{figure}[t]
\centering
\includegraphics[width=3.2in]{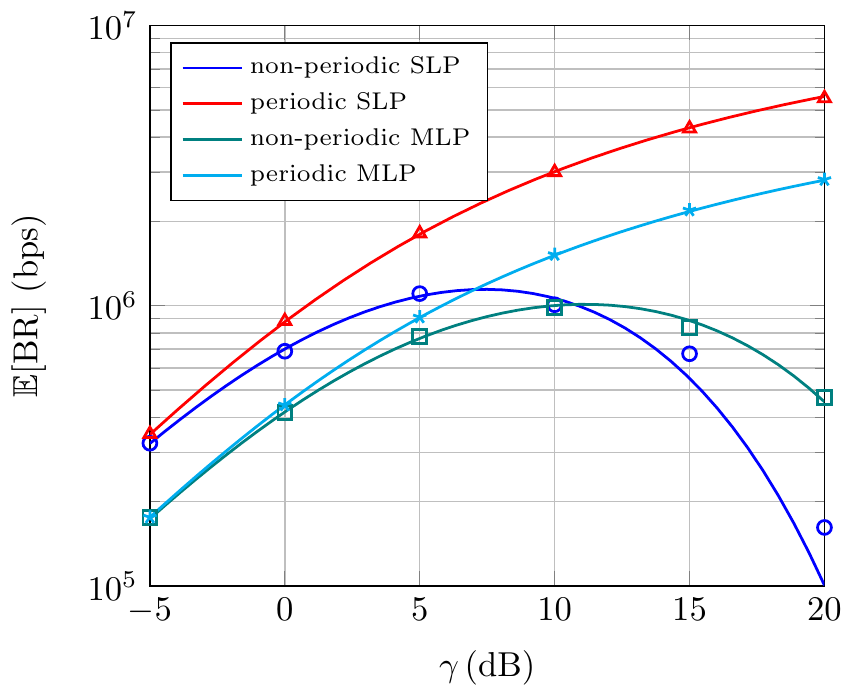}
\caption{Average BR versus the SINR threshold, $\gamma$, for SLP ($d_{\rm safe}=0$ m) and for MLP ($d_{\rm safe}=42$ m) with periodic and non-periodic broadcast messages.}
\label{fig:avBR_gamma}
\end{figure}

\begin{figure}[t]
\centering
\includegraphics[width=\@figSize]{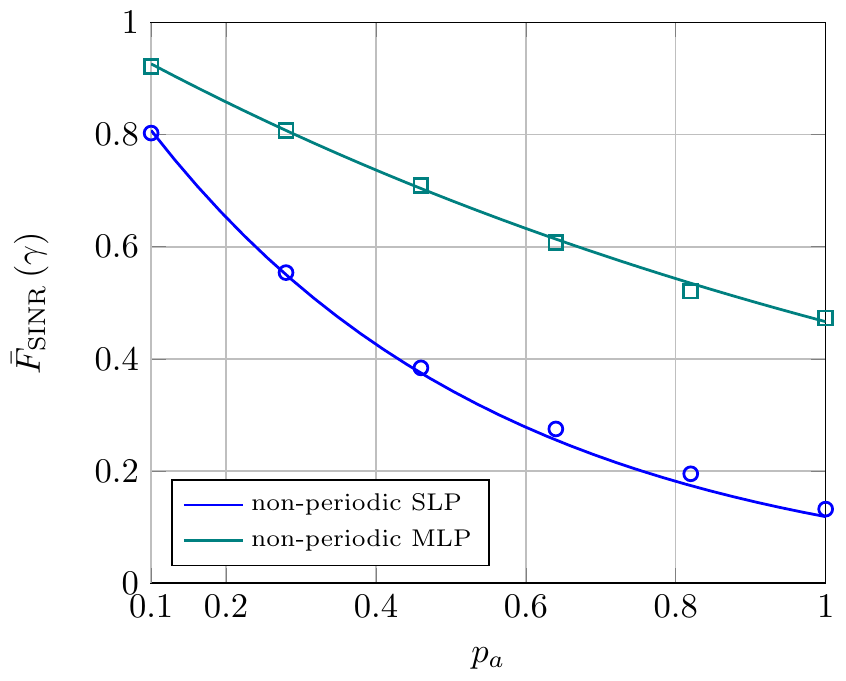}
\caption{Capture probability versus the probability of being active related to non-periodic broadcast messages with SLP and MLP schemes.}
\label{fig:ccdfSINR_pa}
\end{figure}

\begin{figure}[t]
\centering
\includegraphics[width=\@figSize]{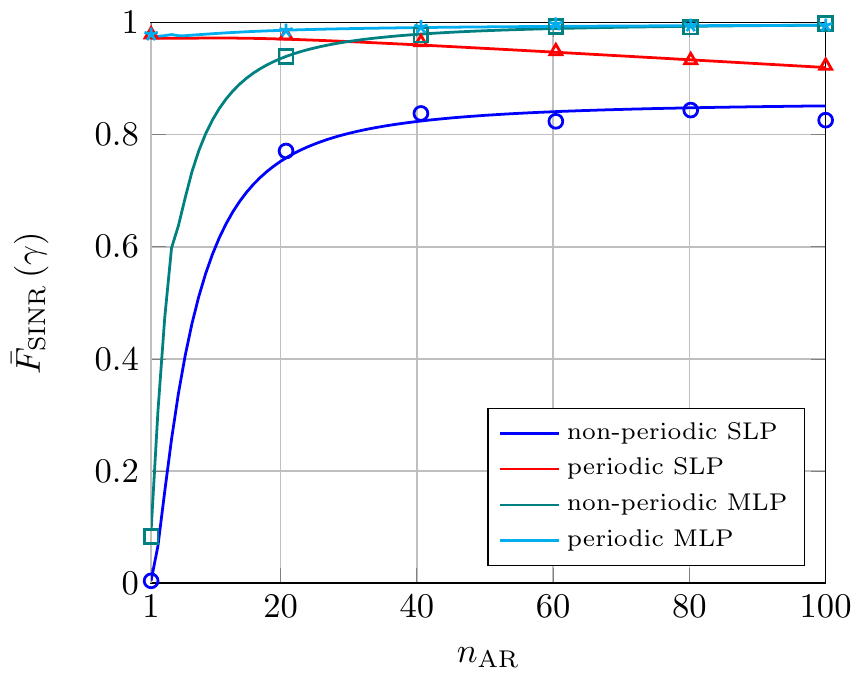}
\caption{Capture probability versus the number of ARs, $n_{\rm AR}$, for SLP ($d_{\rm safe}=0$ m) and for MLP ($d_{\rm safe}=42$ m) with periodic and non-periodic broadcast messages.}
\label{fig:ccdfSINR_nAR}
\end{figure}

In this sub-section, it is evaluated the effect of the density of concurrent transmitters and also, the impact of rate, or equivalently, the SINR threshold, $\gamma$, that leads to the correct reception of a message transmitted with $\log_2 (1+\gamma)$ bps/Hz. Fig. \ref{fig:ccdfSINR_gamma} illustrates the capture probability for SLP and MLP with both periodic and non-periodic messages versus $\gamma$. In the case of non-periodic messages, MLP scheme achieves a higher capture probability than SLP. This is due to the fact that with MLP there is no intra-segment interference, as well as to the fact that the density of interfering vehicles is smaller. If we focus on the case of periodic messages, it is observed a higher capture probability than in the case of non-periodic messages. This is because in this latter case, the probability of being active, which depends on (\ref{eq:pa periodic}), is greatly smaller than in the case of non-periodic messages. According to \textbf{Remark \ref{rem:Density of active vehicles}}, the density of active vehicles in case of periodic reporting is the same in SLP and MLP schemes. However, SLP, contrary to MLP, has intra-segment interference, which explains why SLP has a sightly smaller capture probability, as it can be observed from Fig. \ref{fig:ccdfSINR_gamma}. 

The average BR of correctly received messages versus the SINR threshold is shown in Fig. \ref{fig:avBR_gamma}. Three trends can be observed from the figure. Firstly, it can be noticed that the average BR of periodic messages is higher than non-periodic messages. This is related to the smaller capture probability that exhibits non-periodic messages as it has been discussed above. Secondly, it can be observed that SLP leads to a higher average BR, which is due to the fact that with this scheme the bandwidth per AR is higher. Finally, it is observed that there exist a value of $\gamma$ that maximizes the average BR. 

Fig. \ref{fig:ccdfSINR_pa} illustrates that the capture probability is a decreasing function with respect to $p_a$ as expected, which is related to the higher interference as the density of transmitting vehicles increases. 

\subsection{Impact of the number of Access Resources}
\label{sec:nAR}

\begin{figure}[t]
\centering
\includegraphics[width=\@figSize]{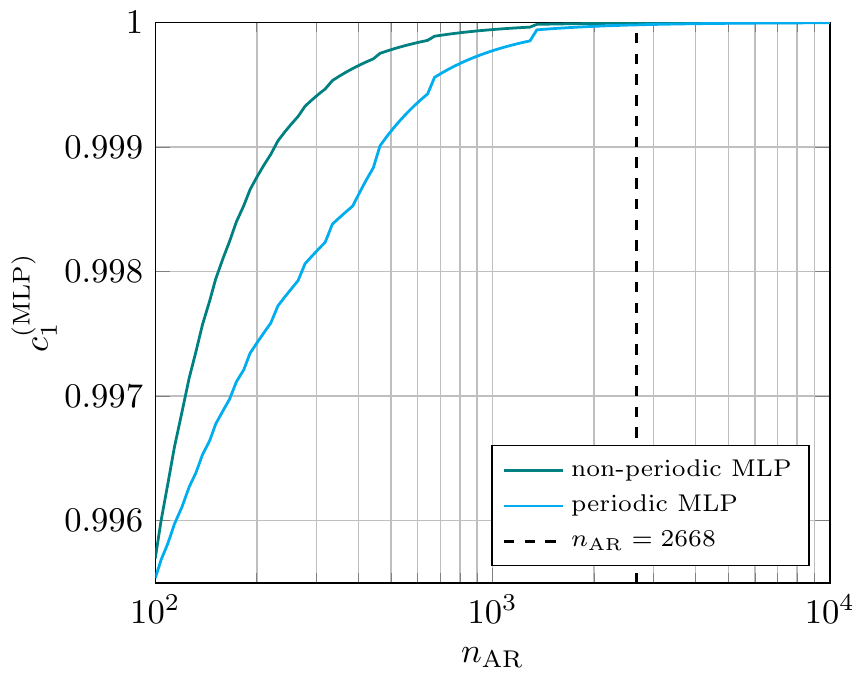}
\caption{Maximal capture probability, $c_1^{(\rm MLP)}$, for $n_{\rm AR}$ ranging from $10^2$ to $10^4$, under MLP ($d_{\rm safe}=42$ m) with periodic and non-periodic broadcast messages with $r_{\rm bc}=150$ m, $d_{\rm max}=56$ km and $\gamma=5$ dB.}
\label{fig:ccdfSINR_nAR_zoom}
\end{figure}


\begin{figure}[t]
\centering
\includegraphics[width=\@figSize]{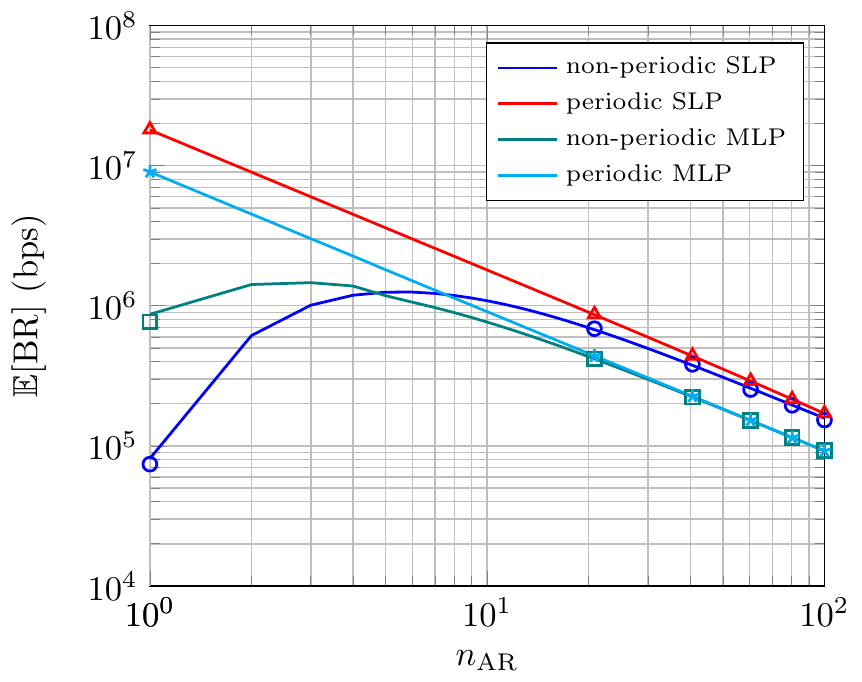}
\caption{Average BR versus the number of ARs, $n_{\rm AR}$, for SLP ($d_{\rm safe}=0$ m) and for MLP ($d_{\rm safe}=42$ m) with periodic and non-periodic broadcast messages.}
\label{fig:avBR_nAR}
\end{figure}


\begin{figure}[t]
\centering
\includegraphics[width=\@figSize]{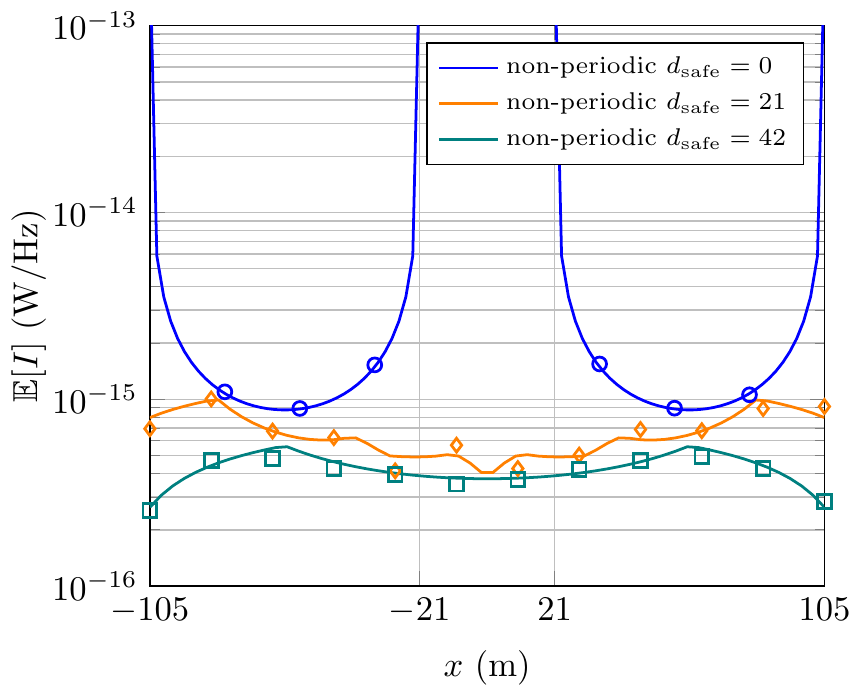}
\caption{Average interference for different locations, $x \in [-105,+105]$ m, for SLP and MLP with $d_{\rm safe}=\{21,42\}$ with non-periodic messages, $\lambda=0.8/84$ vehicles/m and
$n_{\rm AR}=3$.}
\label{fig:avI_x}
\end{figure}

\begin{figure}[t]
\centering
\includegraphics[width=\@figSize]{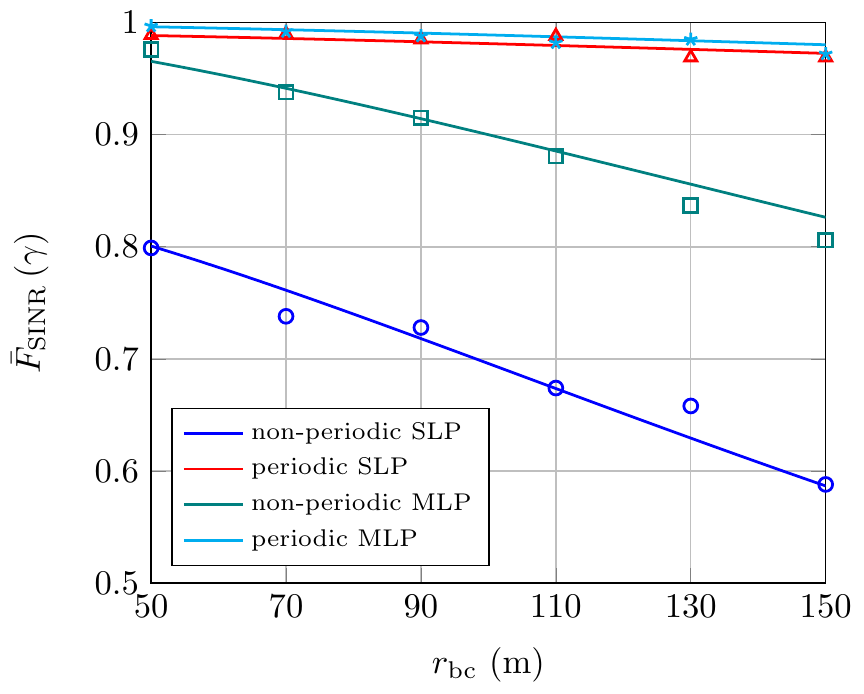}
\caption{Capture probability versus the broadcast distance, $r_{\rm bc}$, for SLP ($d_{\rm safe}=0$ m) and for MLP ($d_{\rm safe}=42$ m) with periodic and non-periodic broadcast messages.}
\label{fig:ccdfSINR_rbc}
\end{figure}

Throughout this sub-section, the impact of the number of ARs is studied. In particular, Fig. \ref{fig:ccdfSINR_nAR} illustrates the capture probability when $n_{\rm AR}$ is ranging from $1$ to $100$. It is observed that in the case of non-periodic messages, with a high active probability ($p_a=0.25$), increasing $n_{\rm AR}$ greatly increases the capture probability for both SLP and MLP. 
This scenario highlights the great potential of GLOC for non-periodic messages where reliability plays a crucial role. 
Such a trend is also observed with periodic messages under MLP; however, with SLP the capture probability of periodic messages is a decreasing function with $n_{\rm AR}$. 
To understand this, let us remark that increasing $n_{\rm AR}$ tend to reduce the interference, since it rises the co-channel distance. On the other hand, increasing $n_{\rm AR}$, leads to an increment of the probability of being active, $p_a$, which also increases the intra-segment interference. The growth of the probability of being active is due to the fact that the time to transmit a periodic message rises with $n_{\rm AR}$. Hence, in case of SLP, intra-segment interference dominates over the capture probability,  which diminishes with $n_{\rm AR}$. 
In case of MLP the capture probability grows with respect to $n_{\rm AR}$, since there is no intra-segment interference, and also because rising $n_{\rm AR}$ increases the distance to interfering vehicles. In particular, with MLP it is obtained reliabilities of $99.55\%$ and $99.76\%$  with non-periodic and periodic messages respectively, at a distance of $150$ m with $n_{\rm AR}=100$. 

V2X communications have to be highly reliable. The exact reliability target changes from standard to standard and depending on the application, e.g.,  $95\%$ (ITS), $99\%$ (LTE V2X) and $99.999\%$ (5G V2X)   \cite{Chen10,5GAA16,3gpp2016}. 
Hence, providing a very high capture probability is a paramount issue. One of the great benefits of MLP is that, according to \textbf{Remark \ref{rem:Noise-limited regime}}, it is possible to move the system from an interference-limited into a noise-limited regime by increasing $n_{\rm AR}$; once the system is noise-limited, the capture probability can be increased by increasing the transmit power as stated in \textbf{Remark \ref{rem: MLP Dependence with noise and transmit power}}. This is illustrated in Fig. \ref{fig:ccdfSINR_nAR_zoom} where it is shown the maximum capture probability, $c_1^{(\rm MLP)}$, which it is achieved for $\rho_{\rm VT} \to \infty$. As stated in \textbf{Remark \ref{rem:Noise-limited regime}}, for our simulation assumptions the system is noise-limited for $n_{\rm AR} > 2668$. 

Finally, Fig. \ref{fig:avBR_nAR} 
illustrates the average BR versus $n_{\rm AR}$. It is shown that, when $10<n_{\rm AR}<100$, SLP achieves a higher average BR than MLP since ARs have more bandwidth in the former case. Regarding non-periodic messages, it is shown that there exist an optimal value of $n_{\rm AR}$, which maximizes the average BR. 
Finally, it is shown that the average BR is a decreasing function with $n_{\rm AR}$, and it tends to $0$ as $n_{\rm AR} \to \infty$, as stated in \textbf{Remark \ref{rem:Average rate when nRB tends to infinity}}.

\subsection{Impact of the broadcast distance and the segment size: $r_{bc}$ and $d_{\mathcal{A}}$}
\label{sec:Spatial design}

\begin{figure}[t]
\centering
\includegraphics[width=\@figSize]{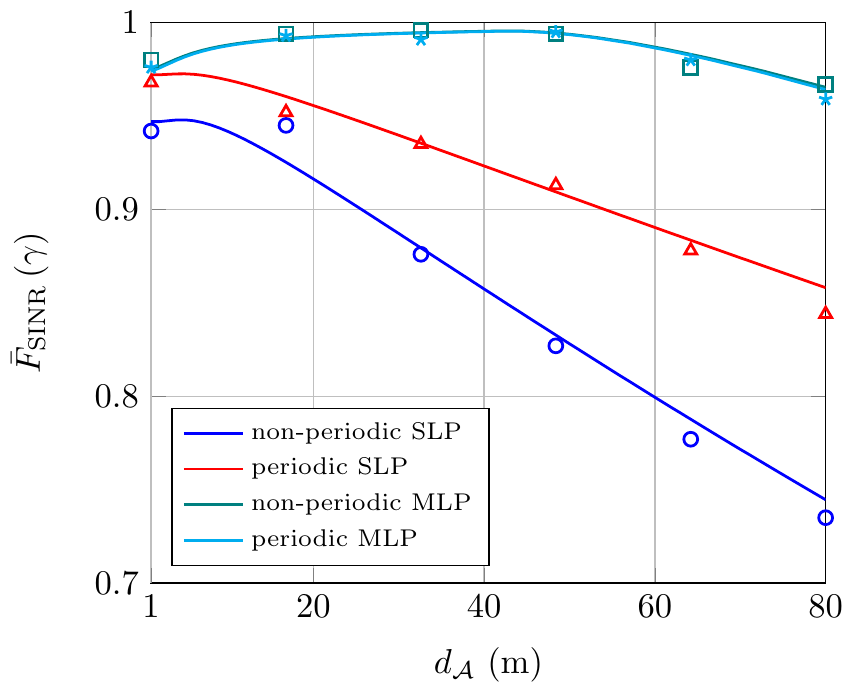}
\caption{Capture probability versus the segment size, $d_\mathcal{A}$, for SLP ($d_{\rm safe}=0$ m) and for MLP ($d_{\rm safe}=42$ m) with periodic and non-periodic broadcast messages.}
\label{fig:ccdfSINR_dA}
\end{figure}

\begin{figure}[t]
\centering
\includegraphics[width=\@figSize]{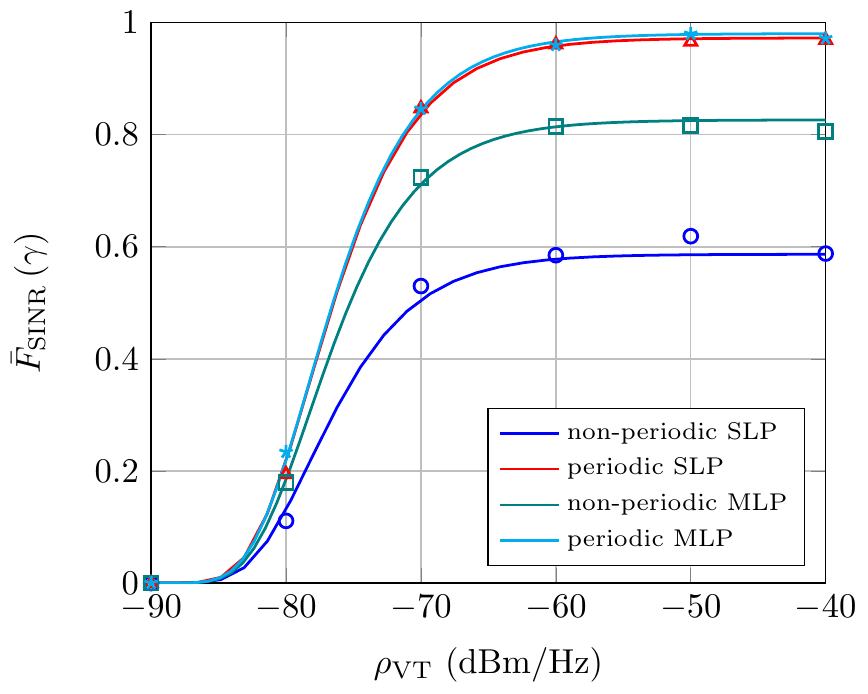}
\caption{Capture probability versus the transmit power per Hz, $\rho_{\rm VT}$, for SLP ($d_{\rm safe}=0$ m) and for MLP ($d_{\rm safe}=42$ m) with periodic and non-periodic broadcast messages.}
\label{fig:ccdfSINR_rhoVT}
\end{figure}


\begin{figure}[t]
\centering
\includegraphics[width=\@figSize]{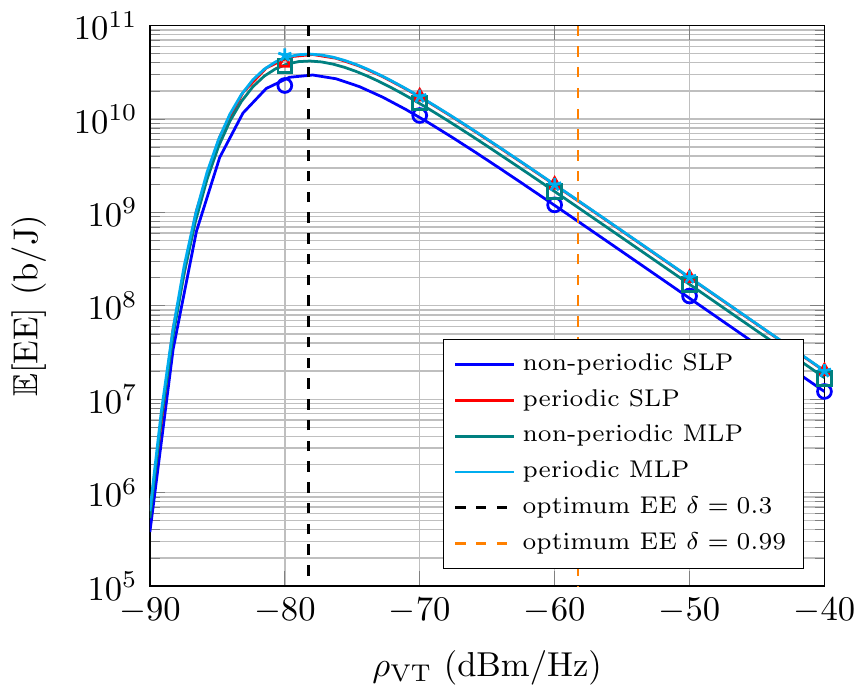}
\caption{Average EE versus the transmit power per Hz, $\rho_{\rm VT}$, for SLP ($d_{\rm safe}=0$ m) and for MLP ($d_{\rm safe}=42$ m) with periodic and non-periodic broadcast messages.}
\label{fig:avEE_rhoVT}
\end{figure}

\begin{figure}[t]
\centering
\includegraphics[width=\@figSize]{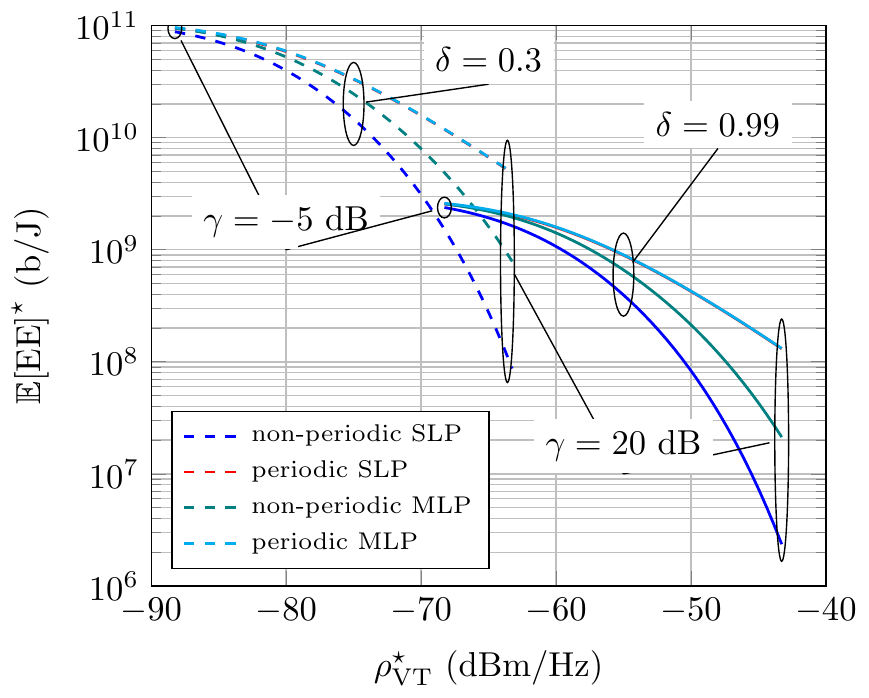}
\caption{Optimum EE versus the optimum transmit power per Hz, $\rho^\star_{\rm VT}$, for SLP ($d_{\rm safe}=0$ m) and for MLP ($d_{\rm safe}=42$ m) with periodic and non-periodic broadcast messages. Each pair of values $(\mathbb{E}[{\rm EE}]^\star,\rho^\star_{\rm VT})$ is obtained for a different SINR threshold, $\gamma$, ranging from $-5$ dB to $20$ dB.}
\label{fig:optAvEE_OptRhoVT}
\end{figure}

The average interference is evaluated in Fig. \ref{fig:avI_x} for locations, $x \in \mathbb{R}$, ranging from $-105$ to 
$+105$ m, where the probe segment is centered at $x=0$. As it is stated in \textbf{Remark \ref{rem:Covergence of the interference}}, it can be observed that without a minimum distance between points, the interference does not converge within co-channel segments, which are centered at multiples of $n_{\rm AR}$. It is used the same density $\lambda$ with both SLP and MLP to assess the effect of the minimum distance between vehicles in the interference. As it is expected, the average interference is reduced as the minimum distance between points, $d_{\rm safe}$, is increased. 

Fig. \ref{fig:ccdfSINR_rbc} shows the capture probability versus the broadcast distance, where it is observed the reduction in capture probability as $r_{\rm bc}$ increases. However, it can be observed that the dependence with $r_{\rm bc}$ is higher in case of non-periodic than in case of periodic messages.

The capture probability versus the segment size is represented in Fig. \ref{fig:ccdfSINR_dA}. To understand this result, two aspects should be taken into account. On the one hand, increasing $d_{\cal A}$ rises the distance to co-channel segments, since this distance depends on $n_{\rm AR} d_{\cal A}$. However, if $d_{\cal A}$ grows the segment-size also rises, which may lead to intra-segment interference.
With SLP, it is shown that the capture probability decreases as $d_{\cal A}$ increases roughly below $10$ m. This is because, since there exist intra-segment interference, the interference rises as $d_{\cal A}$ increases. With MLP, the capture probability grows as $n_{\rm AR}$ rises for $n_{\rm AR} < 42$ m. This is due to the fact that increasing $d_{\cal A}$ rises the distance to co-channel segments. Then, for $n_{\rm AR} > 42$ m, if $n_{\rm AR}$ grows the capture probability decreases since there is now intra-segment interference, which is due to the fact that $d_{\cal A} > d_{\rm safe}$.

\subsection{Transmit Power and Optimum Energy Efficiency}
\label{sec:Transmit Power and Optimum}

\begin{figure}[t]
\centering
\includegraphics[width=\@figSize]{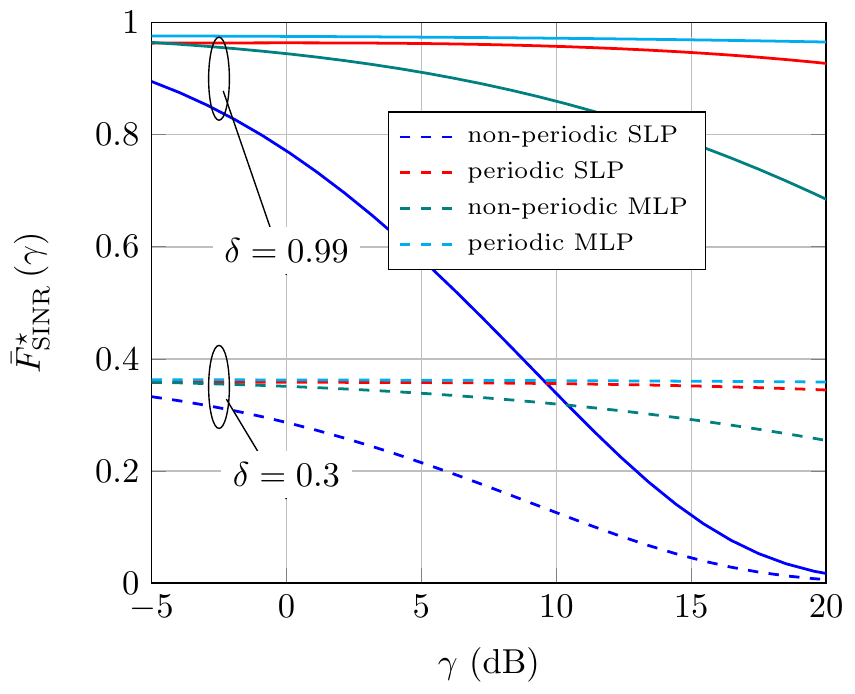}
\caption{Capture probability versus the SINR threshold, $\gamma$, associated with the optimum transmit power, 
$\rho^\star_{\rm VT}$ as appear in Fig \ref{fig:optAvEE_OptRhoVT}.}
\label{fig:optPc_gamma}
\end{figure}

In this sub-section it is shown the effect of the transmit power. Then, its optimal value in terms of EE is studied. 
\textbf{Remark \ref{rem: MLP Dependence with noise and transmit power}} is illustrated in Fig. \ref{fig:ccdfSINR_rhoVT}. It is shown that the capture probability is an increasing function with $\rho_{\rm VT}$, and its maximum value, $c_1^{(k)}$, is roughly achieved when $\rho_{\rm VT}$ is high enough. 

The average EE versus $\rho_{\rm VT}$ is shown in Fig. \ref{fig:avEE_rhoVT}, where it is observed that the global (unconstrained optimum) of the EE is achieved with $\delta=0.3$, as stated in \textbf{Theorem \ref{theo:Optimal transmit power}}. If a higher capture probability, i.e., higher $\delta$, must be satisfied, the optimal transmit power leads to a smaller capture probability, as it is observed for 
$\delta=0.99$ in the figure. Specifically, for MLP with $\delta=0.99$, it is observed that the optimal EE, which is around $1.4 \cdot 10^9$ b/J, is achieved at $\rho_{\rm VT}=-58.27$ dBm/Hz. In that case, the capture probability of periodic and non-periodic messages is $97.41\%$ and $96.27\%$ respectively, but the transmit power is around $20$ dB smaller than in the baseline case ($\rho_{ \rm VT}=-40$ dBm/Hz). 

Fig. \ref{fig:optAvEE_OptRhoVT} illustrates the optimum EE versus the optimal transmit power for different $\gamma$ values ranging from $-5$ to $20$ dBs. Hence, in this figure, every pair of the form $(\mathbb{E}[{\rm EE}]^\star,\rho^\star_{\rm VT})$ is obtained for a different SINR threshold. As it is expected, a smaller $\gamma$ yields to a higher  optimal EE, and a smaller optimal transmit power. 
Additionally, increasing the capture probability reduces the optimal EE and increases the optimal transmit power. 
The capture probability associated with Fig. \ref{fig:optAvEE_OptRhoVT} is illustrated in Fig. \ref{fig:optPc_gamma}. 

\section{Conclusions}
\label{sec:Conclusion}
This paper proposes an abstraction model that allows a tractable analysis and optimization of Geo-Location based access in vehicular networks. 
With such a technique, the road is divided in segments and a single orthogonal Access Resource (AR) is associated with a given segment. The mapping between segments and ARs is made aiming to maximize the co-channel distance. Vehicles determines its corresponding segment based on its geographical position,  therefore reducing the interference when accessing the channel. 
Two frequency allocation schemes are considered: Single-Lane Partition (SLP) and Multi-Lane Partition (MLP). 
MLP distinguish between different lanes and hence it can avoid intra-sement interference, however it uses the bandwidth less efficiently than SLP since it requires orthogonal bandwidth allocations to each lane.
A wide set of analytical results 
are obtained aiming at providing a deep understanding about the proposed schemes. From these analytical results, theoretical insights are derived. 
In particular, it has been shown that: (i) the capture probability is an increasing function with respect to the transmit power with exponential dependence; (ii) the system is noise-limited for MLP when the number of ARs is high enough, whereas it is interference-limited in case of SLP; (iii) the average interference diverges when it is evaluated in co-channel segments with SLP, whereas it always converges for the case of MLP. 
Interestingly, (iii) means that with MLP it is possible to obtain a capture probability as high as necessary, by increasing the number of ARs and transmit power. 
%
This facilitates the implementation of safety applications that requires very high reliability. In particular, with MLP it has been obtained reliabilities of $99.55\%$ and $99.76\%$  with non-periodic and periodic messages respectively, at a distance of $150$ m. 
Finally, the optimum transmit power that achieves maximal EE subject to a minimum capture probability is also obtained. 
Specifically, for MLP it is observed that the optimal EE is achieved with a transmit power of $-58.27$ dBm/Hz which is $20$ dB smaller than baseline case while keeping a $99\%$ of the maximum capture probability.

\appendices

\section{Proof of Lemma \ref{lem:SLP LTI}}
\label{app:Proof of lem SLP LTI}

In case of SLP, the Laplace of the interference can be written as follows 

\begin{align}
\label{eq: LTI_1}
& \mathcal{L}_{I\left( x \right)}\left( s \right)=\mathbb{E}_{I\left( x \right)}\left[ \mathrm{e}^{-sI\left( x \right)} \right] 
\nonumber \\ & \quad \overset{\mathrm{(a)}}{\mathop{=}}\,
\mathbb{E}_{\Phi ^{(\mathrm{a})}}
\Bigg[ \prod\limits_{\mathrm{VT}_{i}\in \Phi ^{(\mathrm{a})} \setminus 
\left\{ \mathrm{VT}_{0} \right\} \bigcap{\mathfrak{b}\left( x,d_{\max } \right)}}\mathbb{E}_{H_{\mathrm{VT}_{i}}}\mathrm{exp 
\Bigg(}-sH_{\mathrm{VT}_{i}} 
\nonumber \\ 
& \quad \times \left( \tau |\mathrm{VT}_{i}-x| \right)^{-\alpha }
\rho_{\mathrm{VT}}\mathbf{1}\left( \mathrm{VT}_{i}\in \mathcal{A}^{(\mathrm{AR}_{0})} 
\right) \Bigg] \Bigg) 
\nonumber \\ & \quad \overset{\mathrm{(b)}}{\mathop{=}}\,
\exp \Bigg(-\lambda \cdot p_{a}\int
\limits_{\mathbb{R}}
{\frac{s\left( \tau |y-x| \right)^{-\alpha }\rho_{\mathrm{VT}}}{1+s\left( \tau |y-x| \right)^{-\alpha } \rho_{\mathrm{VT}}}} 
\nonumber \\ & \quad 
\times \mathbf{1}\left( y \in \mathcal{A}^{(\mathrm{AR}_{0})} \right)
\cdot \mathbf{1}\left( y \in \mathfrak{b}_x \left( d_{\max } \right)  \right) 
\mathrm{d}y \Bigg) 
\nonumber \\ & \quad \overset{\mathrm{(c)}}{\mathop{=}}\,
\exp \Bigg(-\lambda \cdot p_{a}\sum\limits_{c=-\infty }^{\infty }\int\limits_{y=c\cdot n_{\mathrm{AR}}\cdot d_{\mathcal{A}}-\frac{d_{\mathcal{A}}}{2}}^{c\cdot n_{\mathrm{AR}}\cdot d_{\mathcal{A}}+\frac{d_{\mathcal{A}}}{2}} 
\nonumber \\ & \quad 
\frac{s\left( \tau |y-x| \right)^{-\alpha }\rho_{\mathrm{VT}}}{1+s\left( \tau |y-x| \right)^{-\alpha }\rho_{\mathrm{VT}}} \mathbf{1}\left( y\in \mathfrak{b}_x \left( d_{\max } \right) \right)\mathrm{d}y
\Bigg) 
\end{align}

\noindent where (a) comes after expressing exponential of the summation 
that defines the interference as a product over the PPP $\Phi^{(\mathrm{a})}$; 
(b) after applying the PGFL \cite{Haenggi13} of the PPP and performing expectation over the fading
and (c) after expressing the region $\mathcal{A}^{(\mathrm{AR}_0)}$ as a summation of
co-channel segments. Then, we can proceed as appears below

\begin{align}
\label{eq: LTI_2}
& \mathcal{L}_{I\left( x \right)}\left( s \right) 
\overset{\mathrm{(a)}}{\mathop{=}}\,
\exp \Bigg (-\lambda \cdot p_{a}\sum \limits_{c=-\left\lfloor d_{\max }/\left( n_{\mathrm{RB}}d_{\mathcal{A}} \right) \right\rfloor }^{\left\lceil d_{\max }/\left( n_{\mathrm{RB}}d_{\mathcal{A}} \right) \right\rceil }  
\nonumber \\ & \quad 
\int\limits_{t=\max \left( c\cdot n_{\mathrm{RB}}\cdot d_{\mathcal{A}}-\frac{d_{\mathcal{A}}}{2}-x,-d_{\max } \right)}^{\min \left( c\cdot n_{\mathrm{RB}}\cdot d_{\mathcal{A}}+\frac{d_{\mathcal{A}}}{2}-x,d_{\max } \right)}{\frac{s\left( \tau |t| \right)^{-\alpha }\rho_{\mathrm{VT}}}{1+s\left( \tau |t| \right)^{-\alpha }\rho_{\mathrm{VT}}}}\cdot \mathrm{d}t \Bigg)
\nonumber \\ & \quad 
\overset{\mathrm{(b)}}{\mathop{=}}\,
\exp \Bigg(-\lambda \cdot p_{a}\sum\limits_{c=-\left\lfloor d_{\max }/\left( n_{\mathrm{RB}}d_{\mathcal{A}} \right) \right\rfloor }^{\left\lceil d_{\max }/\left( n_{\mathrm{RB}}d_{\mathcal{A}} \right) \right\rceil } \Bigg[ 
\mathbf{1}\left( \mu _{L}^{(1)}<\mu _{U}^{(1)} \right)
\nonumber \\ & \quad 
\int\limits_{t=\mu _{L}^{(1)}}^{\mu _{U}^{(1)}}{\frac{sH_{\mathrm{VT}_{i}}\left( \tau t \right)^{-\alpha }\rho_{\mathrm{VT}}}{1+sH_{\mathrm{VT}_{i}}\left( \tau t \right)^{-\alpha }\rho_{\mathrm{VT}}}}\cdot \mathrm{d}t 
\nonumber \\ & \quad 
+\mathbf{1}\left( \mu_{L}^{(2)}<\mu_{U}^{(2)} \right)
\nonumber \\ & \quad 
\int\limits_{t=\mu _{L}^{(2)}}^{\mu _{U}^{(2)}}{\frac{sH_{\mathrm{VT}_{i}}\left( -\tau t \right)^{-\alpha }\rho_{\mathrm{VT}}}{1+sH_{\mathrm{VT}_{i}}\left( -\tau t \right)^{-\alpha }\rho_{\mathrm{VT}}}}\cdot \mathrm{d}t \Bigg] \Bigg)  
\end{align}

\noindent where (a) comes after applying the maximum distance to
the integration limits and performing the change of variables
$t=y-x$ and (b) comes after expressing the absolute value
function as $|t|=t \cdot \mathbf{1} (t \geq 0) - t \cdot \mathbf{1} (t < 0)$ and applying the
indicator functions to the integration limits. Finally, performing
both integrals and reordering completes the proof.

\section{Proof of Lemma \ref{lem:MLP LTI}}
\label{app:Proof of lem MLP LTI}

The Laplace transform of the interference can be written as

\begin{align}
\label{eq:MLP app LTI 1}
& \mathcal{L}_{I\left( x \right)}\left( s \right)=\mathbb{E}_{I\left( x \right)}\left[ \mathrm{e}^{-sI\left( x \right)} \right] 
\nonumber \\ & \quad 
\overset{\mathrm{(a)}}{\mathop{=}}\,\mathbb{E} \Bigg[ \mathrm{exp} \Bigg(-s\sum\limits_{\mathrm{VT}_{i}\in \Phi ^{(\mathrm{a})}\setminus \left\{ \mathrm{VT}_{\mathrm{0}} \right\}}{H_{\mathrm{VT}_{i}}\left( \tau |\mathrm{VT}_{i}-x| \right)^{-\alpha }\rho _{\mathrm{VT}}} 
\nonumber \\ & \quad
\times \mathbf{1}\left( \mathrm{VT}_{i}\in \mathfrak{b}_{x}\left( d_{\max } \right) \right)\mathbf{1}\left( \mathrm{VT}_{i}\in \mathcal{A}^{(\mathrm{RB}_{0})} \right) 
\nonumber \\ & \quad 
\times \mathbf{1}\left( \left| \mathrm{VT}_{i}-x \right|>d_{\mathrm{safe}} \right)\mathbf{1}\left( \left| \mathrm{VT}_{i}-v \right|>d_{\mathrm{safe}} \right) \Bigg) \Bigg] 
\nonumber \\ & \quad 
\overset{\mathrm{(b)}}{\mathop{=}}\,\exp \Bigg(-\lambda \cdot p_{a}\sum\limits_{c=-\left\lfloor d_{\max }/\left( n_{\mathrm{RB}}d_{\mathcal{A}} \right) \right\rfloor }^{\left\lceil d_{\max }/\left( n_{\mathrm{RB}}d_{\mathcal{A}} \right) \right\rceil }
\nonumber \\ & \quad 
{\int\limits_{t=\max \left( c\cdot n_{\mathrm{RB}}\cdot d_{\mathcal{A}}-\frac{d_{\mathcal{A}}}{2}-x,-d_{\mathrm{max}} \right)}^{\min \left( c\cdot n_{\mathrm{RB}}\cdot d_{\mathcal{A}}+\frac{d_{\mathcal{A}}}{2}-x,d_{\mathrm{max}} \right)}{\frac{s\left( \tau |t| \right)^{-\alpha }\rho _{\mathrm{VT}}}{1+s\left( \tau |t| \right)^{-\alpha }\rho _{\mathrm{VT}}}}} 
\nonumber \\ & \quad 
\times  \mathbf{1}\left( \left| t \right|>d_{\mathrm{safe}} \right)\cdot \mathbf{1}\left( \left| t+x-v \right|>d_{\mathrm{safe}} \right)\cdot \mathrm{d}t \Bigg)
\end{align}

\noindent where (a) comes after applying \textbf{Assumption \ref{ass:HCPP}} and reordering the resulting expression and (b) after applying the PGFL of the PPP, performing expectation over the fading, expressing the region $\mathcal{A}^{(\mathrm{AR}_0)}$ as a summation of
co-channel segments and performing the change of variables
$t=y-x$. Then, we proceed as follows

\begin{align}
\label{eq:MLP app LTI 2}
& \exp \Bigg(-\lambda \cdot p_{a}\sum\limits_{c=-\left\lfloor d_{\max }/\left( n_{\mathrm{RB}}d_{\mathcal{A}} \right) \right\rfloor }^{\left\lceil d_{\max }/\left( n_{\mathrm{RB}}d_{\mathcal{A}} \right) \right\rceil }{{}} 
\nonumber \\ & \quad
\int\limits_{t=\max \left( c\cdot n_{\mathrm{RB}}\cdot d_{\mathcal{A}}-\frac{d_{\mathcal{A}}}{2}-x,-d_{\mathrm{max}} \right)}^{\min \left( c\cdot n_{\mathrm{RB}}\cdot d_{\mathcal{A}}+\frac{d_{\mathcal{A}}}{2}-x,d_{\mathrm{max}} \right)}{\frac{s\left( \tau t \right)^{-\alpha }\rho _{\mathrm{VT}}}{1+s\left( \tau t \right)^{-\alpha }\rho _{\mathrm{VT}}}} 
\nonumber \\ & \quad
\times \mathbf{1}\left( t>d_{\mathrm{safe}} \right)\cdot \mathbf{1}\left( \left| t+x-v \right|>d_{\mathrm{safe}} \right)\cdot \mathrm{d}t 
\nonumber \\ & \quad 
+\int\limits_{t=\max \left( c\cdot n_{\mathrm{RB}}\cdot d_{\mathcal{A}}-\frac{d_{\mathcal{A}}}{2}-x,-d_{\mathrm{max}} \right)}^{\min \left( c\cdot n_{\mathrm{RB}}\cdot d_{\mathcal{A}}+\frac{d_{\mathcal{A}}}{2}-x,d_{\mathrm{max}} \right)}{\frac{s\left( -\tau t \right)^{-\alpha }\rho _{\mathrm{VT}}}{1+s\left( -\tau t \right)^{-\alpha }\rho _{\mathrm{VT}}}} 
\nonumber \\ & \quad
\times \mathbf{1}\left( t<-d_{\mathrm{safe}} \right)\cdot \mathbf{1}\left( \left| t+x-v \right|>d_{\mathrm{safe}} \right)\cdot \mathrm{d}t\Bigg) 
\end{align}

\noindent where it has expressed the absolute value
function as $|t|=t \cdot \mathbf{1} (t \geq 0) - t \cdot \mathbf{1} (t < 0)$ and it has applied the resulting indicator functions to the integration limits. It should be noticed that the following equations holds

\begin{align}
\label{eq:MLP app LTI 3}
& \mathbf{1}\left( \left| t+x-v \right|>d_{\mathrm{safe}} \right)
\nonumber \\ & \quad 
=\mathbf{1}\left( t>d_{\mathrm{safe}}-x+v \right)+\mathbf{1}\left( t<v-x-d_{\mathrm{safe}} \right) 
\nonumber \\ & 
\mathbf{1}\left( \left| t+x-v \right|>d_{\mathrm{safe}} \right)
\nonumber \\ & \quad
=\mathbf{1}\left( t>d_{\mathrm{safe}}-x+v \right)+\mathbf{1}\left( t<v-x-d_{\mathrm{safe}} \right) 
\end{align}

Finally, substituting the above expressions into (\ref{eq:MLP app LTI 2}), performing the resulting integrals and reordering completes the proof.

\ifCLASSOPTIONcaptionsoff
  \newpage
\fi

\bibliographystyle{IEEEtran}
\bibliography{DLOC}

\end{document}